%% file: main.tex
\documentclass[11pt,english]{article}
\usepackage{lmodern}

\usepackage[T1]{fontenc}
\usepackage[utf8]{inputenc}
\usepackage{float}
\usepackage{units}
\usepackage{dsfont}
\usepackage{amsmath}
\usepackage{amsthm}
\usepackage{amssymb}
\usepackage{geometry}
\makeatletter
\numberwithin{figure}{section}
\numberwithin{equation}{section}

\@ifundefined{date}{}{\date{}}
\usepackage{float,psfrag,epsfig,color,url,hyperref}
\usepackage{algorithm,algorithmic}
\usepackage{graphicx,relsize}
\usepackage{amssymb,amsfonts,amsmath,amsthm,amscd,dsfont,mathrsfs,mathtools,microtype,nicefrac,pifont}
\usepackage{upgreek}
\usepackage[dvipsnames]{xcolor}
\usepackage{epstopdf,bbm,enumitem}
\usepackage{dsfont,tikz}
\usepackage[mathscr]{euscript}
\usepackage[toc,page]{appendix}
\hypersetup{
  colorlinks,
  linkcolor={red!50!black},
  citecolor={blue!50!black},
  urlcolor={blue!80!black}
}
\usepackage{etoolbox}
\usepackage{mdframed}
\surroundwithmdframed[
  linewidth=0.5pt, innerleftmargin=8pt, innerrightmargin=8pt,
  innertopmargin=5pt, innerbottommargin=5pt,
  skipabove=\baselineskip, skipbelow=\baselineskip, nobreak=true
]{lyxalgorithm}
\patchcmd{\thebibliography}
  {\settowidth}
  {\setlength{\itemsep}{0pt plus 0.1pt}\settowidth}
  {}{}
\apptocmd{\thebibliography}
  {\small}
  {}{}
\allowdisplaybreaks
\usepackage{custom}
\makeatletter
\renewcommand{\paragraph}{%
  \@startsection{paragraph}{4}%
  {\z@}{1.25ex \@plus 1ex \@minus .2ex}{-1em}%
  {\normalfont\normalsize\bfseries}%
}
\makeatother
\makeatother

\theoremstyle{plain}
\newtheorem{thm}{\protect\theoremname}[section]
\newtheorem{lem}[thm]{\protect\lemmaname}
\theoremstyle{remark}
\newtheorem{rem}[thm]{\protect\remarkname}
\theoremstyle{plain}
\newtheorem{cor}[thm]{\protect\corollaryname}
\newtheorem{prop}[thm]{\protect\propositionname}
\newtheorem{lyxalgorithm}[thm]{\protect\algorithmname}
\usepackage{babel}
\providecommand{\algorithmname}{Algorithm}
\providecommand{\corollaryname}{Corollary}
\providecommand{\lemmaname}{Lemma}
\providecommand{\propositionname}{Proposition}
\providecommand{\remarkname}{Remark}
\providecommand{\theoremname}{Theorem}

\begin{document}
\input{math-macros.tex}

\title{Thin-Shell Stability of Gaussian Cooling:\\
Logconcave Sampling with Sesteric Complexity from a Cold Start \author{Yunbum Kook\\ University of Michigan\\  \texttt{ybkook@umich.edu} \and Santosh S. Vempala\\ Georgia Tech\\ \texttt{vempala@gatech.edu}}}
\maketitle
\begin{abstract}
We show that logconcave probability measures along the Gaussian cooling
path have thin-shell stability, generalizing the thin-shell theorem.
This result leads to improved complexity for the fundamental problem
of sampling an arbitrary logconcave distribution from a cold start.
For (near-)isotropic logconcave distributions, the complexity is nearly
$n^{2.5}$, improving the previous bound of $n^{2.75}$, and matching
the complexity of the abstract Speedy walk. 
\end{abstract}

\section{Introduction}

The study of the complexity of logconcave sampling in high dimension
has led to powerful algorithmic techniques and beautiful geometric
phenomena. It is known that an arbitrary logconcave distribution in
$\Rn$ with density proportional to $e^{-V(x)}$ for some convex function
$V$ can be sampled from a \emph{warm }start\footnote{Roughly speaking, a warm start means that the initial distribution
is already close to the target distribution.} using $\Otilde(n^{2}\cpi\polylog\frac{1}{\veps})$ evaluations of
$V$ \cite{KV26unified}. This is a tight bound for known methods
and conjectured to be the best possible for any algorithm with only
oracle query access to the target density. Here, the Poincar\'e constant
$\cpi$ is conjectured to be a universal constant for isotropic logconcave
distributions (by the celebrated Kannan--Lov\'asz--Simonovits (KLS)
conjecture \cite{KLS95isop}) and known to be bounded by $O(\log n)$
\cite{Klartag23log}.

In this paper, we focus on the problem of logconcave sampling from
a \emph{cold }start (i.e., no warm start), given access to an evaluation
oracle for $V$. For the methods that achieve the best known bounds
from a warm start, such as the $\bw$ \cite{KLS97random} and $\ino$
\cite{KVZ26INO,KV25sampling}, initiating them from an arbitrary point
could lead to very high complexity. This is due to the fact that for
points near the boundary, these methods could make a large number
of attempts (queries) before making a single proper step; moreover,
this is unavoidable. $\har$ has also been shown to have similar complexity
from a warm start \cite{LV06hit,KV26HAR}, and while it has the advantage
of polynomial complexity from an arbitrary interior starting point,
the known analysis incurs an additional factor of dimension, i.e.,
with complexity scaling as $n^{3}$.

Sampling is a key ingredient in algorithms for volume computation
\cite{DFK91random,LS93random,KLS97random,LV06simulated,CV18Gaussian},
integration, and rounding \cite{LV06fast,JLLV26reducing,KV25sampling}.
To give complete guarantees for the latter problems, researchers have
addressed the challenge of generating a warm start efficiently, with
Kannan, Lov\'asz, and Simonovits \cite{KLS97random} giving an algorithm
for convex bodies with complexity roughly $n^{5}$. The main idea
is annealing, i.e., a sequence of distributions, each one providing
a warm start for the next. The starting distribution is easy to sample,
and the final distribution is the target. Lov\'asz and Vempala improved
this complexity to $n^{4}$ \cite{LV06simulated} and extended the
approach to logconcave distributions \cite{LV06fast}. In fact, a
key idea was to use a sequence of logconcave distributions in the
annealing schedule, which provably achieves an improvement over using
a sequence of uniform distributions even for the case of convex bodies
\cite{DFK91random,KLS97random}. 

A decade later, Cousins and Vempala \cite{CV18Gaussian} introduced
$\gc$, an accelerated adaptation of the previous annealing schemes,
and combined it with the $\bw:\msf R_{\infty}\to\tv$\footnote{Recall the $q$-R\'enyi divergence: $\msf R_{q}(\mu\mmid\nu):=\frac{q}{q-1}\log\,\norm{\nicefrac{\D\mu}{\D\nu}}_{L^{q}(\nu)}$
for $q>1$. It is monotonic in $q$, and satisfies $2\,\norm{\cdot}^{2}_{\tv}\leq\KL=\lim_{q\downarrow1}\msf R_{q}\leq\msf R_{2}=\log(1+\chi^{2})\leq\msf R_{\infty}=\log\esssup_{\nu}\nicefrac{\D\mu}{\D\nu}$.
The notation $\msf R_{\infty}\to\tv$ means that the sampler requires
an initial distribution with $\msf R_{\infty}$-warmness and its output
distribution is close to the target in total variation (TV) distance.} to obtain $n^{3}$ complexity for well-rounded distributions\footnote{A distribution is \emph{well-rounded} if a level set of constant measure
contains a ball of constant radius, and the trace of its covariance
is $\Otilde(n)$. Any near-isotropic logconcave distribution, with
covariance close to identity, is well-rounded.}. In particular, this implies that for an arbitrary near-isotropic
distribution, sampling from a cold start has cubic complexity. The
main idea in $\gc$ was an adaptive schedule, which efficiently accelerated
the rate of progress towards the target, and this cubic complexity
remained the state-of-the-art till 2025. We will discuss this in more
detail presently. 

The work of Lee and Vempala \cite{LV18logSoblev,LV24eldan} suggests
that further improvement might be possible. They showed that the\emph{
$\sw$}, an abstract process introduced by \cite{KLS97random} and
used only for analysis, has iteration complexity $n^{2.5}$ (not query
complexity) and that this bound is tight. The $\sw$ counts only proper
steps of the $\bw$ and not wasted steps, which are needed by any
known implementation; whether this $n^{2.5}$ iteration complexity
can be matched by an actual algorithm in oracle-query complexity has
remained an elusive question.

For uniform sampling, Kook, Vempala, and Zhang \cite{KVZ24INO,KVZ26INO}
developed the $\ino$ algorithm (also known as the proximal sampler
\cite{LST21structured}), using algorithmic diffusion and obtaining
a $\msf R_{\infty}\to\msf R_{q}$ guarantee. Beyond its stronger output
guarantee, the algorithmic connection to heat flow makes tools from
Markov semigroup theory, functional inequalities, and convex geometry
available for its analysis. Subsequently, Kook and Zhang \cite{KZ25Renyi}
refined the analysis of $\gc$ to obtain $\msf R_{\infty}$-warmness
with the same cubic complexity, but warm-start generation was still
a bottleneck toward subcubic complexity from a cold start. $\ino$
and the annealing algorithm were extended to general logconcave distributions
through the exponential lifting technique \cite{KV25sampling}, matching
the complexity of uniform sampling. 

These ideas were seminal for the first improvement in cold start sampling
since \cite{CV18Gaussian}: using a more careful annealing schedule,
refining the bound on the log-Sobolev constant of logconcave probability
measures with compact support, and relying on the stronger and faster
end-to-end sampling guarantee in $q$-R\'enyi divergence, \cite{KV25faster,KV26zeroLC}
improved the complexity of cold start sampling to $n^{2.75}$, the
first subcubic bound. They also gave an example showing a lower bound
of $n^{8/3}$ for their general method. The motivation for the present
paper is to understand the best possible complexity of sampling from
a cold start. 

Our main algorithmic result here is a refinement of $\gc$, which,
together with $\ino$, achieves the complexity of nearly $n^{2.5}$
for arbitrary near-isotropic logconcave distributions. In other words,
it matches the known iteration complexity of the abstract $\sw$ without
any overhead. The core idea of the proof is a new stability phenomenon:
while the Poincar\'e constant of logconcave distributions along the
$\gc$ path can vary substantially (and thus using it, as in previous
work, leads to higher complexity), the \emph{thin-shell constant }is
stable along this path, and this suffices! We will discuss this in
more detail shortly, after presenting the main results precisely. 

\subsection{Results\label{sec:setup}}

\paragraph{Setup.}

Let $V:\R^{n}\to\R\cup\{\infty\}$ be a convex function such that
$\D\pi(x)\propto e^{-V(x)}\,\D x$ is a full-dimensional probability
measure in $\Rn$. Following the convention in \cite{KV25sampling},
let us denote its ground set by $\msf L_{g}=\{x\in\R^{n}:V(x)-\inf V\le10n\}$,
and standardize the setup by assuming $B(0,1)\subset\msf L_{g}$ (via
translation and scaling). For the uniform distribution on a convex
body $\K\subset\Rn$, this corresponds to $B(0,1)\subseteq\K$. We
assume access to its zeroth-order oracle (i.e., evaluation oracle
for $V$ telling us $V(x)$ for a queried point $x$).

We denote the barycenter by $m:=\E_{\pi}X$, the covariance matrix
by $\Sigma:=\cov_{\pi}X=\cov\pi$, and its operator norm by $\Lambda:=\norm{\cov\pi}_{\op}$.
Since the barycenter $m$ is not necessarily at the origin, we also
define $L:=\norm{\Sigma+mm^{\T}}_{\op}$ and $R^{2}=\E_{\pi}[\norm{\cdot}^{2}]$.
Note that $\Lambda\leq L,\tr\Sigma\leq R^{2}$ in general. For $t>0$,
we define a Gaussian tilt of $\pi$ as $\pi_{t}(\D x)\propto\exp(-\frac{\norm x^{2}}{2t})\,\pi(\D x)$,
which can be thought of as the continuous version of $\gc$. The stability
proofs themselves use only logconcavity and these parameters, whereas
the ground-set normalization is needed only for the sampling results.

Before we proceed, we define the \emph{anisotropic thin-shell constant}
$\cQ_{n}$ of logconcave measures as the smallest number such that
every isotropic logconcave distribution $\nu$ in $\Rn$ and every
symmetric matrix $M$ satisfy
\begin{equation}
\var_{\nu}(Z^{\T}MZ)\le\cQ_{n}\tr(M^{2})\,.\label{eq:QF}
\end{equation}
The ordinary thin-shell conjecture \cite{ABP03central,BK03central},
which was proven by Klartag and Lehec \cite{KL25thin} in 2025, posits
$\msf Q_{n}=O(1)$ when $M=I_{n}$. Recently, Chen and Klartag showed
that $\msf Q_{n}=8$ is sharp when $M=I_{n}$ \cite{CK26thinshell}.
For symmetric $M$, it holds that $\msf Q_{n}\lesssim\log n$ due
to Klartag's KLS bound~\cite{Klartag23log}. 

\paragraph{Result 1: Thin-shell stability along Gaussian tilts.}

Consider the family of Gaussian tilts $\{\pi_{t}\}_{t\geq0}$ that
interpolates between $\pi_{0}=\delta_{0}$ and $\pi_{\infty}=\pi$.
The first result shows that the variance of $\norm X^{2}$ can be
bounded in terms of the barycenter and covariance matrix of the base
measure $\pi$.
\begin{thm}
[Stability of thin-shell estimates]\label{thm:main-stability} Let
$\pi$ be any logconcave probability measure in $\Rn$ with $m=\E_{\pi}X$
and $\Sigma=\cov\pi$, and $\pi_{t}$ be its Gaussian tilt for $t>0$.
Then,
\[
\sup_{t>0}\var_{\pi_{t}}(\norm X^{2})\lesssim(1+\cQ_{n})^{2}\,\bbrack{\tr(\Sigma^{2})+\norm m^{2}\,\bpar{\tr(\Sigma^{2})+m^{\T}\Sigma m}^{1/2}}=\Otilde(R^{2}L\wedge R^{3}\Lambda^{1/2})\,.
\]
\end{thm}

In particular, when $\pi$ is isotropic, this result implies that
the thin-shell constant remains bounded: $\frac{1}{n}\var_{\pi_{t}}(\norm X^{2})\lesssim(1+\msf Q_{n})^{2}$.
Note that this extends the thin-shell bound \cite{KL25thin,CK26thinshell}
to \emph{all} $t$ (up to $\log$ factors). Its stronger version would
be stability of the covariance operator norm, but this is false in
general \cite{Bizeul26logsobolev,KV26zeroLC}; see \S\ref{subsec:Thin-shell-stability-under}
for more details.

\paragraph{Result 2: Faster warm-start generation for logconcave distributions.}

We start with warm-start generation for the uniform distribution on
a convex body. Using the thin-shell stability result, our algorithm
follows the $\gc$ scheme of \cite{KV26zeroLC} with more aggressive
updates to $\sigma^{2}$. Roughly speaking, it updates $\sigma^{2}$
by a factor of $(1+\sigma^{2}/\msf V^{1/2})$ in each phase, where
$\msf V$ is the thin-shell bound on $\sup_{t>0}\var_{\pi_{t}}(\norm X^{2})$.
Using $\ino$ for truncated Gaussians, whose query complexity is $n^{2}\sigma^{2}$
from an $O(1)$-warm start in $\msf R_{q}$, doubling a given $\sigma^{2}$
requires $n^{2}\sigma^{2}\times\msf V^{1/2}/\sigma^{2}=n^{2}\msf V^{1/2}$
queries. Since there are logarithmically many doublings, the total
query complexity for warm-start generation is $\Otilde(n^{2}\msf V^{1/2})$.
This is simply $\Otilde(n^{2.5})$ for isotropic uniform distributions.
Hence, we ``algorithmically'' achieve the iteration complexity of
the $\sw$ in \cite{LV24eldan}.
\begin{thm}
[Faster warm start for uniform sampling]\label{thm:warm-uniform}
Let $\K\subset\R^{n}$ be a convex body with $B(0,1)\subset\K$ and
$\pi$ uniform on $\K$. Algorithm~\ref{alg:uniform-annealing} outputs
a law $\nu$ satisfying $\msf R_{q}(\nu\mmid\pi)\le1$ for $q\geq1$,
using $\Otilde(q^{1/2}n^{2}\msf V^{1/2})=\Otilde\bpar{q^{1/2}n^{2}\min\{RL^{1/2},R^{3/2}\Lambda^{1/4}\}}$
expected membership queries. When $\pi$ is nearly isotropic (i.e.,
$\cov\pi\approx I_{n}$), the total complexity is $\Otilde(q^{1/2}n^{2.5})$.
\end{thm}

This improves the previously best complexity \cite{KV25faster,KV26zeroLC}.
It is conceivable that $n^{2.5}$ is the best possible complexity
for cold-start sampling in light of the $n^{2.5}$ lower bound for
the iteration complexity of the $\sw$ \cite{LV24eldan}. 

For general logconcave distributions, we use the exponential lift
\cite{KV25sampling}, streamline the tilted Gaussian annealing in
\cite{KV26zeroLC}, and accelerate it using thin-shell stability.
\begin{thm}
[Faster warm start for logconcave sampling]\label{thm:warm-main}
Let $V:\Rn\to\R\cup\{\infty\}$ be a convex function, and suppose
that $\pi\propto e^{-V}$ satisfies $B(0,1)\subseteq\msf L_{g}$.
Given access to an evaluation oracle for $V$, Algorithm~\ref{alg:lifted-annealing}
outputs a law $\nu$ satisfying $\msf R_{q}(\nu\mmid\pi)\le2$ for
$q\geq1$, using $\Otilde(n^{2.5}+q^{1/2}n^{2}\msf V^{1/2})$ expected
evaluation queries, which can be bounded as
\[
\Otilde\bpar{n^{2.5}+q^{1/2}n^{2}\min\{RL^{1/2},R^{3/2}\Lambda^{1/4}\}}\,.
\]
When $\pi$ is nearly isotropic (i.e., $\cov\pi\approx I_{n}$), the
total complexity is $\Otilde(q^{1/2}n^{2.5})$.
\end{thm}

\subsection{Technical overview}

Our proofs are organized around a direct connection between $\gc$
and thin-shell stability. We first recall the intuition behind $\gc$
\cite{CV15bypass,CV18Gaussian}, focusing on how thin-shell stability
determines its annealing rate. We then outline the proof of thin-shell
stability and explain how it leads to a faster annealing schedule
for logconcave distributions.

\paragraph{Gaussian cooling and variance of thin shells.}

Let $\pi$ be a logconcave distribution on $\Rn$, and consider the
Gaussian-tilted family $\pi_{\sigma^{2}}(\D x)\propto\exp(-\frac{\norm x^{2}}{2\sigma^{2}})\,\pi(\D x)$.
$\gc$ gradually increases $\sigma^{2}$, thereby weakening the effect
of Gaussian tilts and interpolating toward the target $\pi$. To understand
the annealing rate, we take a geometric view of $\gc$. The fluctuation
scale of $\norm X^{2}$ is $\var_{\pi_{\sigma^{2}}}(\norm X^{2})^{1/2}$.
Geometrically, when we increase $\sigma^{2}$, the next annealing
distribution should overlap substantially with the thin shell of the
previous annealing distribution $\pi_{\sigma^{2}}$ so that consecutive
distributions are ``close'' enough for a sampler to mix rapidly toward
the next distribution. Hence, heuristically, $\E_{\pi_{\sigma^{2}_{\mathrm{new}}}}[\norm X^{2}]$
should lie in $\E_{\pi_{\sigma^{2}}}[\norm X^{2}]\pm O(\var_{\pi_{\sigma^{2}}}(\norm X^{2})^{1/2})$.

\begin{figure}[H]
\centering
\begin{centering}
\begin{tikzpicture}[x=1cm,y=1cm,>=stealth,
    line cap=round,line join=round,font=\small,
    guide/.style={line width=0.45pt,dash pattern=on 3pt off 3pt},
    profile/.style={line width=1.15pt},
    measure/.style={<->,line width=0.85pt}]
\definecolor{coolingblue}{RGB}{30,105,220}
% Squared-radius profiles: schematic, as in the original sketch.
\draw[profile,black]
    (0.45,0.18)
    .. controls (1.90,0.33) and (2.60,1.06) .. (3.55,1.58)
    .. controls (4.60,2.19) and (5.45,2.55) .. (6.50,2.55)
    .. controls (7.40,2.55) and (8.20,2.50) .. (9.15,2.13)
    .. controls (10.00,1.80) and (10.65,1.08) .. (11.25,0.67)
    .. controls (12.00,0.23) and (12.55,0.16) .. (13.00,0.12);
\draw[profile,coolingblue]
    (2.70,0.07)
    .. controls (4.20,0.04) and (5.00,0.43) .. (6.25,1.22)
    .. controls (7.45,1.97) and (8.55,2.82) .. (9.65,2.82)
    .. controls (10.60,2.82) and (11.20,2.16) .. (12.00,1.31)
    .. controls (12.90,0.39) and (13.50,0.13) .. (14.15,0.07);
% Mean and fluctuation-scale guides.
\draw[guide,black!65] (2.50,0) -- (2.50,0.91);
\draw[guide,black!65] (6.50,-1.34) -- (6.50,2.55);
\draw[guide,black!65] (10.70,-1.34) -- (10.70,1.13);
\draw[guide,coolingblue!75] (9.65,0) -- (9.65,2.82);
\draw[->,line width=0.8pt] (0.15,0) -- (14.55,0);
\node[anchor=north east] at (14.55,-0.12) {$\norm X^2$};
\node[anchor=south] at (3.05,2.76)
    {$\operatorname{Law}_{\pi_{\sigma^2}}(\norm X^2)$};
\node[anchor=south,text=coolingblue] at (12.05,2.91)
    {$\operatorname{Law}_{\pi_{(1+\alpha)\sigma^2}}(\norm X^2)$};
\node[anchor=north,fill=white,inner sep=3pt] at (6.50,-0.13)
    {$\E_{\pi_{\sigma^2}}[\norm X^2]$};
\node[anchor=north,text=coolingblue,fill=white,inner sep=3pt]
    at (9.65,-0.13) {$\E_{\pi_{(1+\alpha)\sigma^2}}[\norm X^2]$};
% Mean displacement and one-sided fluctuation scale.
\draw[measure,coolingblue] (6.56,0.53) -- (9.59,0.53);
\node[anchor=south,text=coolingblue,inner sep=2pt] at (8.075,0.67)
    {$\displaystyle \frac{\alpha}{2\sigma^2}\,\var_{\pi_{\sigma^2}}(\norm X^2)$};
\draw[measure,black] (6.56,-1.10) -- (10.64,-1.10);
\node[anchor=north,inner sep=3pt] at (8.60,-1.25)
    {$\Theta\bpar{\sqrt{\var_{\pi_{\sigma^2}}(\norm X^2)}}$};
\end{tikzpicture}
\par\end{centering}
\caption{$\protect\gc$ viewed through the distribution of the squared radius.
The schematic profiles show the laws of $\protect\norm X^{2}$ under
consecutive Gaussian tilts. Increasing $\sigma^{2}$ to $(1+\alpha)\,\sigma^{2}$
shifts the mean roughly by the amount indicated by the blue arrow.
Comparing this displacement with the current fluctuation scale (black
arrow) motivates the annealing step size.\label{fig:gaussian-cooling-shells}}
\end{figure}
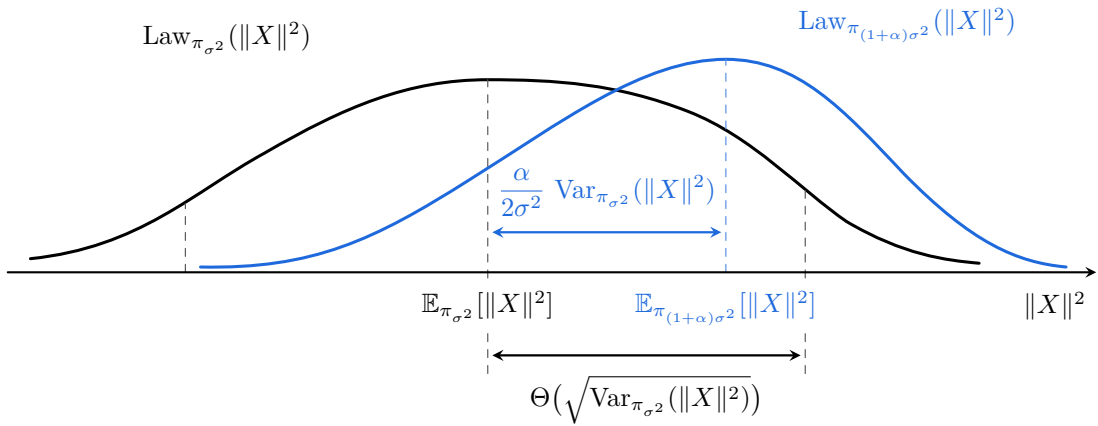

To make this intuition concrete, let $\sigma^{2}_{\mathrm{new}}=(1+\alpha)\,\sigma^{2}$
for a small increment $\alpha$. Direct computation gives $\de_{\sigma^{2}}\E_{\pi_{\sigma^{2}}}[\norm X^{2}]=\var_{\pi_{\sigma^{2}}}(\norm X^{2})/(2\sigma^{4})$.
Thus, the center of the thin shell moves roughly by $\frac{\alpha}{2\sigma^{2}}\var_{\pi_{\sigma^{2}}}(\norm X^{2})$.
Since we want this change to be smaller than the annulus width $O(\var_{\pi_{\sigma^{2}}}(\norm X^{2})^{1/2})$,
this suggests the multiplicative step $\alpha\lesssim\sigma^{2}/\var_{\pi_{\sigma^{2}}}(\norm X^{2})^{1/2}$.
Note that a larger variance has two competing effects: it gives a
larger annulus, but it moves its center even faster, so $\sigma^{2}/\var_{\pi_{\sigma^{2}}}(\norm X^{2})^{1/2}$
is the ideal balance between these two effects.

This identifies $\var_{\pi_{\sigma^{2}}}(\norm X^{2})$ as a fundamental
quantity underlying $\gc$ and motivates bounding it uniformly. Previous
work bounds it by $\var_{\pi_{\sigma^{2}}}(\norm X^{2})\leq4\sigma^{2}\,\E_{\pi_{\sigma^{2}}}[\norm X^{2}]$
using the strong logconcavity of $\pi_{\sigma^{2}}$ and the Poincar\'e
inequality (or Brascamp--Lieb). Since $\E_{\pi_{\sigma^{2}}}[\norm X^{2}]$
can be made $O(R^{2})$ by a suitable truncation, the previous bound
yields $\alpha\asymp\sigma/R$. On the other hand, a uniform bound
$\sup_{\sigma^{2}>0}\var_{\pi_{\sigma^{2}}}(\norm X^{2})\leq\msf V$
suggests $\alpha\asymp\sigma^{2}/\msf V^{1/2}$. Once $\sigma^{2}\geq1$,
this suggests a faster annealing schedule than those in previous work.
Theorem~\ref{thm:main-stability} establishes such a $\sigma^{2}$-independent
uniform bound along $\gc$. 

\subsubsection{Thin-shell stability of Gaussian tilts\label{subsec:Thin-shell-stability-under}}

To build intuition, consider an isotropic logconcave probability measure
$\nu$ on $\Rn$, and let $\D\nu_{s}\propto e^{-s\,\norm x^{2}}\,\D\nu$
denote its Gaussian tilt for $s\in[0,\infty]$, where we use the inverse
parameter $s=1/(2t)$ for the sake of exposition. Prior work \cite{KV25faster}
proposed a stronger version of the thin-shell stability conjecture
as part of an attempt to design a faster warm-start algorithm, based
on geometric intuition for stability of the largest eigenvalue of
$\cov\nu_{s}$. Indeed, by \eqref{eq:pi},
\[
\var_{\nu_{s}}(\norm X^{2})\leq4\cpi(\nu_{s})\,\E_{\nu_{s}}[\norm X^{2}]\lesssim_{\log n}\norm{\cov\nu_{s}}_{\op}\,\E_{\nu}[\norm X^{2}]=n\,\norm{\cov\nu_{s}}_{\op}\,,
\]
where the second inequality follows from $\cpi(\nu_{s})\lesssim\norm{\cov\nu_{s}}_{\op}\log n$
(i.e., the logarithmic KLS bound) and $\E_{\nu_{s}}[\norm X^{2}]\le\E_{\nu}[\norm X^{2}]$
(which follows by differentiating $\E_{\nu_{s}}[\norm X^{2}]$ with
respect to $s$). Following the same reasoning, one might conjecture
that $\lda_{s}:=\norm{\cov\nu_{s}}_{\op}$ should be $O(\norm{\cov\nu}_{\op})=O(1)$.
However, it was later disproved by \cite{KV26zeroLC} using Bizeul's
example \cite{Bizeul26logsobolev}. 

A natural alternative is to directly control $\msf V_{s}:=\var_{\nu_{s}}(\norm X^{2})$
without going through \eqref{eq:pi}. Hence, we study the thin-shell
variance $\msf V_{s}$ as $s$ decreases from $\infty$ to $0$. Two
competing factors determine $\msf V_{s}$: (i) the base measure $\nu$
and (ii) the variance $\frac{1}{2s}$ of the Gaussian factor. When
$s$ is large (so the Gaussian variance is small), the Gaussian factor
pulls the $\nu$-mass toward the origin. In this regime, the effect
of the base measure $\nu$ should be limited. Hence, $\nu_{s}$ is
concentrated near the origin, and intuitively, $\msf V_{s}$ should
be small. As $s$ decreases (so the Gaussian variance is larger),
we expect the $\nu_{s}$-measure to become more dispersed, which should
in turn increase $\msf V_{s}$. However, when $s$ is sufficiently
small, the Gaussian weight is $\Theta(1)$ throughout the effective
support of $\nu$, which contains most of the $\nu$-measure (e.g.,
$\norm x=O(n^{1/2})$). In this regime, the Gaussian tilt has little
effect. Hence, $\nu_{s}$ should be distributed similar to the isotropic
measure $\nu$, and we expect $\msf V_{s}\approx\var_{\nu}(\norm X^{2})$.
By the thin-shell theorem, we expect $\msf V_{s}=O(n)$. In summary,
as $s$ decreases from $\infty$ to $0$, the thin-shell variance
should initially increase as the Gaussian factor weakens. Once the
effect of the original isotropic measure becomes dominant, however,
$\msf V_{s}$ should remain $O(n)$. Thus, a plausible conjecture
is that $\msf V_{s}=O(n)$ for \textbf{all} $s\in[0,\infty]$.

\paragraph{Decomposition of variance.}

For a general logconcave measure $\pi$ with $m=\E X$ and $\Sigma=\cov\pi$,
let $Z=\Sigma^{-1/2}(X-m)$ for $X\sim\pi$, and let $\nu$ be the
law of $Z$, which is isotropic and logconcave. Under this change
of variables, the tilted law $\pi_{t}$ becomes
\[
\nu_{s}(\D z)\propto\exp(-s\,\norm{\Sigma^{1/2}z+m}^{2})\,\nu(\D z)\qquad\text{for }s=(2t)^{-1}\,,
\]
and $\norm X^{2}=\norm{\Sigma^{1/2}Z+m}^{2}$. Hence, to bound $\var_{\pi_{t}}(\norm X^{2})$,
we consider the following problem: for $M\succ0$ and $u\in\Rn$,
define $q_{u}(x)=\norm{M^{1/2}x+u}^{2}$ under the tilt $\nu_{s}(\D x)\propto e^{-sq_{u}(x)}\,\nu(\D x)$
for $s\in[0,\infty]$ (later we can simply take $M=\Sigma$ and $u=m$).
Write $b_{s}$ and $B_{s}$ for the mean and covariance of $\nu_{s}$,
and set $y_{s}=M^{1/2}b_{s}+u$ and $M_{s}=M^{1/2}B_{s}M^{1/2}$.
Isotropizing $\nu_{s}$ and applying the definition of $\cQ_{n}$
gives
\[
\var_{\nu_{s}}q_{u}\lesssim\cQ_{n}\tr(M^{2}_{s})+y^{\T}_{s}M_{s}y_{s}\,.
\]
Hence, it suffices to control $\tr(M^{2}_{s})$ and the contribution
of the moving mean $y_{s}$. The Brascamp--Lieb inequality yields
$M_{s}\preceq(2s)^{-1}\,\Id$, but applying this bound independently
in every direction loses a factor of $n$ and becomes vacuous as $s\downarrow0$.
As mentioned earlier, the plausible alternative of proving $\norm{B_{s}}_{\op}\lesssim\norm{B_{0}}_{\op}$
is also false: covariance operator-norm stability under Gaussian tilting
is false in general. This suggests controlling the eigenvalues \emph{collectively},
rather than the largest eigenvalue.

\paragraph{Collective spectral control of $M_{s}$.}

Consider first the centered case $u=0$. Let $\lambda_{i}(s)$ be
the eigenvalues of $M_{s}$, and let $P_{s}$ and $N_{s}$ denote
the total increase and decrease of these eigenvalues relative to $M_{0}=M$:
\[
N_{s}:=\sum_{\lambda_{i}(s)<\lambda_{i}(0)}\bpar{\lambda_{i}(0)-\lambda_{i}(s)}\qquad\text{and}\qquad P_{s}:=\sum_{\lambda_{i}(s)>\lambda_{i}(0)}\bpar{\lambda_{i}(s)-\lambda_{i}(0)}\,.
\]
Note that $\tr M_{s}-\tr M_{0}=P_{s}-N_{s}$. The score identity (Lemma~\ref{lem:score-identity})
gives $\de_{s}\E_{s}[X^{\T}MX]=-\var_{\nu_{s}}q_{u}$, which relates
$N_{s}$, $P_{s}$, and $\norm{y_{s}}^{2}$ as follows: $\Delta_{s}:=\int^{s}_{0}\var_{\nu_{r}}q_{u}\,\D r=N_{s}-P_{s}-\norm{y_{s}}^{2}$,
so $P_{s}+\norm{y_{s}}^{2}\leq N_{s}$. Also, differentiating $\lda_{i}(s)$
along the same path gives $\partial_{s}\log\lambda_{i}(s)=-\cov_{s}(Z^{2}_{s,i},q_{u})$,
where $Z_{s}$ is the random variable for isotropization of $\nu_{s}$.
The log-spectrum estimate yields $N^{2}_{s}\leq s\,\cQ_{n}\tr(M^{2})\,\Delta_{s}$.
Since $\Delta_{s}\leq N_{s}$, both spectral gain and squared mean
displacement are $O(s)$. Combining these relations yields $\tr(M^{2}_{s})\leq\tr(M^{2})\,(1+\msf Q_{n})$
and $y^{\T}_{s}M_{s}y_{s}\leq\half\,\tr(M^{2})\,\msf Q_{n}$ with
$\msf Q_{n}\lesssim\log n$ (Lemma~\ref{lem:property-ms-ys}). Note
that we avoid the na\"ive factor-$n$ loss: instead of using the
same worst-case bound on every eigenvalue, we control the $\ell_{2}$-norm
of the entire log-spectrum collectively. The detailed proof appears
in \S\ref{sec:thin-shell-unshifted}. 

\paragraph{Extension to the shifted case.}

When $u\neq0$, the same argument starts from $y_{0}=u$, and the
identity becomes $\Delta_{s}=N_{s}-P_{s}+\norm u^{2}-\norm{y_{s}}^{2}$.
This change adds $O(\norm u^{2}/s)$ to $\var_{\nu_{s}}q_{u}$, which
is useful but vacuous near the untilted endpoint (i.e., $s=0$). At
that endpoint, the same decomposition and the quadratic-form estimate
directly bound $\var_{\nu}q_{u}$ in terms of $G:=\tr(M^{2})+u^{\T}Mu$.
Moreover, the score identity and reverse H\"older inequality for
quadratic polynomials give $\abs{\partial_{s}\var_{\nu_{s}}q_{u}}\lesssim(\var_{\nu_{s}}q_{u})^{3/2}$.
Solving this Gr\"onwall-type inequality, we use the resulting bound
up to $s_{0}\asymp[(\cQ_{n}+4)G]^{-1/2}$ and the spectral estimate
thereafter. Concatenating the two regimes yields the additional term
$\norm u^{2}\sqrt{G}$ in Theorem~\ref{thm:anisotropic}. Putting
these together, we obtain the bound $\msf V=\Otilde(R^{2}L\wedge R^{3}\Lambda^{1/2})$.

\subsubsection{From variance stability to faster Gaussian cooling}

\paragraph{R\'enyi-divergence between annealing distributions.}

The heuristic calculation above suggests the natural scale of $\alpha\asymp\sigma^{2}/\msf{\sqrt{V}}$.
To make this quantitative, let $F(s)$ be the logarithm of the normalization
constant of $\nu_{s}(\D x)\propto e^{-s\,\norm x^{2}}\,\pi(\D x)$.
Differentiation gives $F''(s)=\var_{\nu_{s}}(\norm X^{2})$. By the
definition of $q$-R\'enyi divergence, one can check that $\msf R_{q}(\nu_{s}\mmid\nu_{s'})=[F(s+(q-1)\,\Delta)-qF(s)+(q-1)\,F(s-\Delta)]/(q-1)$
for $\Delta:=s-s'$. Using Taylor expansion, for $I=[s',s'+q\Delta]$,
we have 
\[
F(s+(q-1)\,\Delta)\le F(s)+(q-1)\,\Delta F'(s)+\frac{V_{I}}{2}\,(q-1)^{2}\Delta^{2}\ \mbox{ and }\ F(s-\Delta)\le F(s)-\Delta F'(s)+\frac{V_{I}}{2}\,\Delta^{2},
\]
so 
\[
\msf R_{q}(\nu_{s}\mmid\nu_{s'})\leq\frac{q}{2}\,(s-s')^{2}\sup_{t\in I}\var_{\nu_{t}}(\norm X^{2})\,.
\]
If $s=(2\sigma^{2})^{-1}$, consecutive targets remain $O(1)$-close
in $\msf R_{q}$ under the update $\sigma^{2}\gets\sigma^{2}(1+\Theta(\sigma^{2}/\sqrt{q\msf V}))$.
The same inequality shows that $\sigma^{2}\gtrsim\sqrt{q\msf V}$
is sufficient for warmness with respect to the target $\pi$, so we
set $\sigma^{2}_{\text{last}}\asymp\sqrt{q\msf V}$. As discussed
in Remark~\ref{rem:comparison-closeness}, the resulting bound recovers
prior closeness bounds in \cite{CV18Gaussian,KV25faster,KV26zeroLC}.
Prior works use the following bound: 
\[
\var_{\nu_{s}}(\norm X^{2})\leq4\sigma^{2}\E_{\nu_{s}}[\norm X^{2}]\lesssim\sigma^{2}R^{2}
\]
by \eqref{eq:pi} (or Brascamp--Lieb); this $\sigma^{2}$-dependent
bound becomes loose as $\sigma^{2}$ increases.

\paragraph{Special case: uniform sampling.}

For a convex body containing the unit ball, we can sample the initial
target at $\sigma^{2}_{0}=1/n$ exactly by Gaussian rejection sampling.
We then follow this schedule using $\ino$ (the proximal sampler)
for truncated Gaussians \cite{KV26zeroLC}. At scale $\sigma^{2}$,
each phase (i.e., sampling initialized at an $O(1)$-warm start in
$\msf R_{q}$) requires $\Otilde(n^{2}\sigma^{2})$ queries, and doubling
the initial $\sigma^{2}$ requires $O(\sqrt{q\msf V}/\sigma^{2})$
phases. Thus, each doubling costs $\Otilde(n^{2}\sqrt{q\msf V})$
queries. Since there are $O(\log(qn\msf V))$ doublings, by the stability
estimate, the total complexity of Gaussian annealing is $\Otilde(q^{1/2}n^{2}\min\{R\sqrt{L},R^{3/2}\Lambda^{1/4}\})$
queries; for a nearly isotropic target, this is $\Otilde(q^{1/2}n^{5/2})$.
Using hypercontractivity of the proximal sampler~\cite{KV26zeroLC},
this procedure propagates $\msf R_{q}$-warmness from one target to
the next. The details are given in \S\ref{sec:uniform}.

\paragraph{General case: logconcave sampling.}

Given a convex potential $V$ and an evaluation oracle for it, we
use the lifting technique \cite{KV25sampling}; instead of sampling
directly from $\pi\propto e^{-V}$, we sample from a density proportional
to $e^{-nt}\,\ind_{\K}(x,t)$, where $\K=\{(x,t):V(x)\leq nt\}$ is
the epigraph of $V/n$. Its $X$-marginal is exactly proportional
to $e^{-V(x)}$. Suitable truncation $\bar{\K}$ of $\K$ contains
at least half of the lifted measure and ensures a finite log-Sobolev
constant. Since the truncated $X$-marginal has density at most twice
that of $\pi$, its parameters $R,L,\Lambda$ are bounded by universal
multiples of those of $\pi$. Hence, the required thin-shell upper
bound increases by at most a universal factor.

On the truncated lift, we use the two-parameter annealing distribution
from \cite{KV25sampling}:
\[
\mu_{\sigma^{2},\rho}(\D x,\D t)\propto\exp\bpar{-\frac{\norm x^{2}}{2\sigma^{2}}-\rho t}\,\ind_{\bar{\K}}(x,t)\,\D x\D t\,.
\]
Using the proximal sampler for this distribution \cite{KV26zeroLC},
we follow the path $(n^{-1},1)\to(n^{-1},n)\to(\sigma^{2}_{\text{last}},n)$
in the $(\sigma^{2},\rho)$-plane, streamlining the annealing schedules
proposed in \cite{KV25faster,KV26zeroLC}. The first segment changes
only $\rho$ while keeping the strongly logconcave Gaussian component
fixed. Once $\rho=n$, the $X$-marginal is exactly the Gaussian tilt
of the truncated target, so the second segment can reuse the schedule
and theory developed for uniform sampling.

For the first phase $(n^{-1},1)\to(n^{-1},n)$, we establish a new
closeness bound. As in the analysis of the $\msf R_{q}$ bound between
consecutive annealing distributions, we have to bound $\var_{(\sigma^{2},\rho)}T:=\var_{\mu_{\sigma^{2},\rho}}T$.
Varentropy of the joint logconcave density \cite{FMW16varentropy},
together with the Poincar\'e bound for its strongly logconcave $X$-marginal,
gives $\var_{(n^{-1},\rho)}T\lesssim n/\rho^{2}$ for $1\leq\rho\leq n$.
Hence, for $q_{0}=\widetilde{\Theta}(1)$, we multiply $\rho$ by
$1+\Theta(1/\sqrt{q_{0}\,(q_{0}+n)})$ at each phase. There are $\Otilde(\sqrt{n})$
such phases, and each costs $\Otilde(n^{2})$ queries (since the proximal
sampler costs $\Otilde(n^{2}(\sigma^{2}\vee1))$ from an $O(1)$-warm
start), using $\Otilde(n^{5/2})$ queries. We maintain the same base
order $q_{0}$ while increasing $\sigma^{2}$ to $1$. Thereafter,
using the hypercontractivity of the proximal sampler again (to boost
the R\'enyi order from $q_{0}$ to $q$), we use the thin-shell update
up to $\sigma^{2}_{\text{last}}$. These Gaussian-annealing phases
cost $\Otilde(n^{2}\sqrt{q\msf V})$ queries as in the uniform case.
Finally, data processing under $(x,t)\mapsto x$ and the constant-mass
truncation transfer the lifted warm start to the original target.

\subsection{Preliminaries\label{subsec:Preliminaries}}

We use the same symbol for a distribution and its density. Both $a\lesssim b$
and $a=\O(b)$ mean $a\le cb$ for a universal constant $c>0$. $a=\Omega(b)$
means $a\gtrsim b$, and $a\asymp b$ means $a=\O(b)$ and $a=\Omega(b)$.
Lastly, $a=\Otilde(b)$ means $a=O(b\polylog b)$. For a positive
semidefinite matrix $\Sigma$, $\norm{\Sigma}_{\op}$ denotes the
operator norm of $\Sigma$. We use $a\vee b$ and $a\wedge b$ to
denote their maximum and minimum, respectively. $B(x,r)$ denotes
an $\ell_{2}$-ball of radius $r$ centered at $x$.

\paragraph{Logconcavity.}

We call a function $f:\Rn\to[0,\infty)$ \emph{logconcave} if $-\log f$
is convex in $\Rn$, and call a probability measure $\pi$ (or distribution)
logconcave if it has a logconcave density with respect to Lebesgue
measure. We assume that all logconcave distributions considered in
this work are full-dimensional\footnote{Otherwise, we may work on the affine subspace supporting the distribution.}.
For $t\geq0$, $\pi$ is called \emph{$t$-strongly logconcave} if
$-\log\pi$ is $t$-strongly convex (i.e., $-\log\pi-\frac{t}{2}\,\abs{\cdot}^{2}$
is convex). Logconcavity is preserved under multiplication, and a
classical result by Pr\'ekopa and Leindler ensures that convolution
also preserves logconcavity. A distribution is called \emph{isotropic}
if its barycenter is at the origin and its covariance matrix is the
identity. Logconcave distributions have exponentially decaying tails
and hence finite moments of all orders. 

\paragraph{Probability.}

Let $\mu$ and $\nu$ denote probability measures on $\Rn$ with $\mu\ll\nu$.
For $q>1$, we define the warmness convention $M_{q}=\norm{\D\mu/\D\nu}_{L^{q}(\nu)}$.
The \emph{$q$-R\'enyi divergence} is defined as $\msf R_{q}(\mu\mmid\nu):=\frac{1}{q-1}\,\log\norm{\frac{\D\mu}{\D\nu}}^{q}_{L^{q}(\nu)}$,
and the \emph{R\'enyi-infinity divergence} is $\msf R_{\infty}(\mu\mmid\nu):=\log\esssup_{\nu}\frac{\D\mu}{\D\nu}$.
Recall monotonicity in the order: $\msf R_{q}\leq\msf R_{q'}$ for
$q\leq q'$. A weak triangle inequality holds for R\'enyi divergence:
for any $q>1$ and probability measures $\mu,\nu,\pi$, it holds that
\[
\msf R_{q}(\mu\mmid\pi)\leq\frac{2q}{2q-1}\,\msf R_{2q-1}(\mu\mmid\nu)+\msf R_{2q}(\nu\mmid\pi)\,.
\]
We refer to \cite{vH14renyi,Mironov17renyi} for more properties of
R\'enyi divergence.

\paragraph{Functional inequalities.}

A probability measure $\pi$ on $\Rn$ is said to satisfy a \emph{Poincar\'e
inequality }with constant $C$ if for any locally Lipschitz function
$f\in L^{2}(\pi)$,
\begin{equation}
\var_{\pi}f:=\int\Bpar{f-\int f\,\D\pi}^{2}\,\D\pi\leq C\int\abs{\nabla f}^{2}\,\D\pi\,,\tag{\ensuremath{\msf{PI}}}\label{eq:pi}
\end{equation}
and the smallest such $C$ is called the Poincar\'e constant $\cpi(\pi)$.
A probability measure $\pi$ on $\Rn$ is said to satisfy a \emph{logarithmic
Sobolev inequality} with constant $C$ if for any locally Lipschitz
function $f\in L^{2}(\pi)$,
\begin{equation}
\ent_{\pi}(f^{2}):=\int f^{2}\log f^{2}\,\D\pi-\int f^{2}\,\D\pi\cdot\log\int f^{2}\,\D\pi\leq2C\int\abs{\nabla f}^{2}\,\D\pi\,,\tag{\ensuremath{\msf{LSI}}}\label{eq:lsi}
\end{equation}
and the smallest such $C$ is referred to as the log-Sobolev constant
$\clsi(\pi)$. In general, \eqref{eq:lsi} is strictly stronger than
\eqref{eq:pi}. By the Bakry--\'Emery criterion, every $t$-strongly
logconcave probability measure $\pi$ satisfies $\clsi(\pi)\leq1/t$
\cite{BGL14analysis}.

\paragraph{Historical note and AI use.}

In early June, during the conference ``Scalable MCMC Sampling'' held
at the FIM--Institute for Mathematical Research at ETH Z\"urich,
the authors realized how to handle the $n^{8/3}$ lower bound example
from \cite{KV25faster} in $n^{2.5}$ oracle complexity. They formulated
and discussed the thin-shell stability conjecture and promising approaches
to proving it. The first author then used GPT 5.5 and 5.6 Pro to obtain
a preliminary version of the proof of Theorem 1.1 through guided prompting.
The authors verified, streamlined, and wrote the full proof and its
consequences.

\section{Thin-shell stability of Gaussian tilts\label{sec:stability}}

In this section, we show that the thin-shell variance remains bounded
for Gaussian tilts; motivated by \eqref{eq:QF}, a natural anisotropic
version of the conjecture would be $\var_{\nu_{s}}(X^{\T}MX)\lesssim_{\log n}\tr(M^{2})$
for every $s\geq0$. In \S\ref{sec:thin-shell-unshifted}, we provide
the proof of this conjecture, and in \S\ref{sec:shifted-tilts},
we provide the following generalization to shifted quadratic functions:
\begin{thm}
[Shifted anisotropic quadratic tilt]\label{thm:anisotropic} Let
$\nu$ be an isotropic logconcave probability measure on $\R^{n}$,
$M\succeq0$, and $u\in\R^{n}$. Set $q_{u}(x)=\norm{M^{1/2}x+u}^{2}$
and define the probability measure $\nu_{s}(\D x)\propto e^{-sq_{u}(x)}\,\nu(\D x)$
for $s\geq0$. Then, for every $s\geq0$,
\[
\var_{\nu_{s}}q_{u}\lesssim(1+\msf Q_{n})^{2}\,\bpar{\tr(M^{2})+\norm u^{2}\sqrt{\tr(M^{2})+u^{\T}Mu}}\,.
\]
\end{thm}

Using this result, we readily obtain the thin-shell stability result
(Theorem~\ref{thm:main-stability}) needed for faster $\gc$.

\subsection{Proof idea: the centered case\label{sec:thin-shell-unshifted}}

We first present the proof in the centered case of Theorem~\ref{thm:anisotropic}
(i.e., $u=0$).
\begin{thm}
[Centered anisotropic quadratic tilt]\label{thm:centered-anisotropic}
Let $\nu$ be an isotropic logconcave probability measure on $\R^{n}$,
and $\nu_{s}(\D x)\propto e^{-s\,\inner{x,Mx}}\,\nu(\D x)$ for $s\geq0$
and $M\succeq0$. Then,
\begin{equation}
\var_{\nu_{s}}(X^{\T}MX)\leq(2\cQ^{2}_{n}+6\cQ_{n})\tr(M^{2})\qquad\text{for every }s\geq0\,.\label{eq:centered-anisotropic}
\end{equation}
\end{thm}

We first prove the theorem for $M\succ0$. Denote the barycenter and
covariance of $\nu_{s}$ by $b_{s}=\E_{s}X$ and $B_{s}=\cov_{s}X$.
Under $x\mapsto y=T(x):=M^{1/2}x$, the pushforward $Y_{s}\sim T_{\#}\nu_{s}$
has barycenter $y_{s}=M^{1/2}b_{s}$ and covariance $M_{s}=M^{1/2}B_{s}M^{1/2}$.
Isotropy of $\nu$ gives $b_{0}=0$ and $B_{0}=\Id$, so $y_{0}=0$
and $M_{0}=M$. Although $\nu$ is centered, the quadratic tilt $\nu_{s}$
need not preserve its mean.

\paragraph{Controlling the covariance and the moving mean.}

Isotropize $\nu_{s}$ by $Z:=M^{-1/2}_{s}M^{1/2}(X-b_{s})$ for $X\sim\nu_{s}$.
We now bound the variance of the target quadratic form $q(x):=x^{\T}Mx$.
We have the decomposition
\[
q(x)=x^{\T}Mx=(x-b_{s})^{\T}M(x-b_{s})+2b^{\T}_{s}Mx-b^{\T}_{s}Mb_{s}=z^{\T}M_{s}z+2y^{\T}_{s}M^{1/2}_{s}z+b^{\T}_{s}Mb_{s}\,.
\]
Hence, $\E_{s}q=\tr M_{s}+\norm{y_{s}}^{2}$ and $q-\E_{s}q=(Z^{\T}M_{s}Z-\tr M_{s})+2y^{\T}_{s}M^{1/2}_{s}Z.$
Since $\var_{s}q=\norm{q-\E_{s}q}^{2}_{L^{2}(\nu_{s})}$, the triangle
inequality in $L^{2}$ gives
\begin{equation}
\sqrt{\var_{s}q}\leq\norm{z^{\T}M_{s}z-\tr M_{s}}_{L^{2}(\nu_{s})}+2\,\norm{y^{\T}_{s}M^{1/2}_{s}z}_{L^{2}(\nu_{s})}\leq\sqrt{\cQ_{n}\tr(M^{2}_{s})}+2\sqrt{y^{\T}_{s}M_{s}y_{s}}\,.\label{eq:unshifted-target}
\end{equation}
Here, \eqref{eq:QF} bounds the first term by $[\cQ_{n}\tr(M^{2}_{s})]^{1/2}$,
while $\norm{y^{\T}_{s}M^{1/2}_{s}z}^{2}_{L^{2}(\nu_{s})}=\E_{s}[(y^{\T}_{s}M^{1/2}_{s}Z)^{2}]=y^{\T}_{s}M_{s}y_{s}$.
Therefore, the proof reduces to controlling (i) the \emph{covariance
contribution} $\tr(M^{2}_{s})$ and (ii) the \emph{moving-mean contribution}
$y^{\T}_{s}M_{s}y_{s}$. 

We now bound these two terms. By Brascamp--Lieb, $M_{s}\preceq(2s)^{-1}\Id$,
so each eigenvalue $(\lda_{i}(s))_{i\in[n]}$ of $M_{s}$ is bounded
by $(2s)^{-1}$. However, applying this bound directly to $\tr(M^{2}_{s})$
would lose a factor of the dimension $n$ and would blow up as $s\downarrow0$.
To refine the analysis, we track changes in the eigenvalues \emph{collectively}
relative to $M_{0}=M$ and relate their total increase to their total
decrease, rather than applying the same worst-case bound to all $n$
eigenvalues. To make this precise, let us define the total spectral
``loss'' and ``gain'' by
\[
N_{s}:=\sum_{\lambda_{i}(s)<\lambda_{i}(0)}\bpar{\lambda_{i}(0)-\lambda_{i}(s)}\qquad\text{and}\qquad P_{s}:=\sum_{\lambda_{i}(s)>\lambda_{i}(0)}\bpar{\lambda_{i}(s)-\lambda_{i}(0)}\,.
\]
Note that $\tr M_{s}-\tr M_{0}=P_{s}-N_{s}$.

We first track how much $\tr(M^{2}_{s})$ has changed from $\tr(M^{2}_{0})$:
since $\tr(M^{2}_{s})-\tr(M^{2}_{0})=\sum_{i}(\lambda_{i}(s)^{2}-\lambda_{i}(0)^{2})$,
\[
\tr(M^{2}_{s})-\tr(M^{2}_{0})\leq\sum_{\lda_{i}(s)>\lda_{i}(0)}\bpar{\lambda_{i}(s)^{2}-\lambda_{i}(0)^{2}}\underset{(i)}{\leq}\frac{1}{s}\sum_{\lambda_{i}(s)>\lambda_{i}(0)}\bpar{\lambda_{i}(s)-\lambda_{i}(0)}=\frac{P_{s}}{s}\,,
\]
where $(i)$ follows from $\lambda_{i}(s)+\lambda_{i}(0)\leq2\lambda_{i}(s)\leq s^{-1}$
for an eigenvalue with $\lambda_{i}(s)>\lambda_{i}(0)$. Thus, the
two target terms are bounded as 
\begin{equation}
\tr(M^{2}_{s})\leq\tr(M^{2}_{0})+\frac{P_{s}}{s}\qquad\text{and}\qquad y^{\T}_{s}M_{s}y_{s}\leq\frac{\norm{y_{s}}^{2}}{2s}\,.\label{eq:final-bound}
\end{equation}
Hence, it suffices to bound $P_{s}$ and $\norm{y_{s}}^{2}$ in terms
of $s$. An elementary identity couples $P_{s}$, $N_{s}$, and $\norm{y_{s}}^{2}$.

\paragraph{Coupling $P_{s}$, $N_{s}$, and $\protect\norm{y_{s}}^{2}$.}

We first recall the following identity.
\begin{lem}
\label{lem:score-identity} Let $\nu$ be a probability measure, $q:\Rn\to\R$
measurable, and define an exponential tilt $\nu_{s}(\D x)\propto e^{-sq(x)}\,\nu(\D x)$
on an interval where the tilt is well-defined. If $f_{s}$ is differentiable
and differentiation can be passed under the defining integrals, then
$\de_{s}\E_{s}f_{s}=\E_{s}[\partial_{s}f_{s}]-\cov_{s}(f_{s},q)$,
where the identity is interpreted entrywise when $f_{s}$ is vector-
or matrix-valued.
\end{lem}

\begin{proof}
Recall an elementary property of the score function: $\E_{s}[\partial_{s}\log\nu_{s}]=\int\partial_{s}\nu_{s}=0$
due to $\partial_{s}\nu_{s}=\nu_{s}\,\partial_{s}\log\nu_{s}$. Differentiating
the expectation under the integral gives
\[
\de_{s}\E_{s}f_{s}=\int(\partial_{s}f_{s})\,\nu_{s}+\int f_{s}\,\partial_{s}\nu_{s}=\E_{s}[\partial_{s}f_{s}]+\E_{s}[f_{s}\,\partial_{s}\log\nu_{s}]=\E_{s}[\partial_{s}f_{s}]+\cov_{s}(f_{s},\de_{s}\log\nu_{s})\,,
\]
where the last equality follows from $\E_{s}[\partial_{s}\log\nu_{s}]=0$.
The claim follows from $\de_{s}\log\nu_{s}=-(q-\E_{s}q)$.
\end{proof}

Using this, $\de_{s}\E_{s}q=-\var_{\nu_{s}}q$. Since $\E_{s}q=\tr M_{s}+\norm{y_{s}}^{2}$
and $\E_{0}q=\tr M_{0}$,
\begin{equation}
\Delta_{s}:=\int^{s}_{0}\var_{\nu_{r}}q\,\D r=\E_{0}q-\E_{s}q=N_{s}-P_{s}-\norm{y_{s}}^{2}\,.\label{eq:unshifted-energy}
\end{equation}
Since $\Delta_{s}\geq0$ (as $\var_{\nu_{r}}q\ge0$), we have $P_{s}+\norm{y_{s}}^{2}\leq N_{s}$.
It remains to bound $N_{s}$ in terms of $s$.

For an index with $\lambda_{i}(s)<\lambda_{i}(0)$, the elementary
inequality $1-r\leq-\log r$ for $0<r\leq1$ gives $\lambda_{i}(0)-\lambda_{i}(s)\leq\lambda_{i}(0)\,\abs{\log\frac{\lambda_{i}(s)}{\lambda_{i}(0)}}$.
By Cauchy--Schwarz over the eigenvalues, 
\begin{equation}
N_{s}\leq\Bpar{\sum_{i}\lambda_{i}(0)^{2}}^{1/2}\Bpar{\sum_{i}\log^{2}\frac{\lambda_{i}(s)}{\lambda_{i}(0)}}^{1/2}\leq\bpar{\tr(M^{2}_{0})}^{1/2}\Bpar{\sum_{i}\log^{2}\frac{\lambda_{i}(s)}{\lambda_{i}(0)}}^{1/2}\,.\label{eq:unshifted-deficit}
\end{equation}
To obtain an explicit factor of $s$ on the right-hand side, we apply
Cauchy--Schwarz in time:
\begin{equation}
\sum_{i}\log^{2}\frac{\lambda_{i}(s)}{\lambda_{i}(0)}=\sum_{i}\Bpar{\int^{s}_{0}\partial_{r}\log\frac{\lambda_{i}(r)}{\lambda_{i}(0)}\,\D r}^{2}\leq s\int^{s}_{0}\sum_{i}\bpar{\partial_{r}\log\lambda_{i}(r)}^{2}\,\D r\,.\label{eq:sum-log2}
\end{equation}

\paragraph{Spectral analysis of $M_{s}$.}

We now analyze how fast each $\lda_{i}$ changes:
\begin{lem}
Let $u_{i}=u_{i}(s)$ be a suitable unit eigenvector of $M_{s}$ such
that $M_{s}u_{i}=\lda_{i}(s)\,u_{i}$, and let $Z_{i}:=u^{\T}_{i}Z$.
Then,
\begin{equation}
\partial_{s}\log\lambda_{i}(s)=-\cov_{s}(Z^{2}_{i},q)\,.\label{eq:unshifted-eigenvalue-derivative}
\end{equation}
\end{lem}

\begin{proof}
By Lemma~\ref{lem:score-identity} for the centered covariance,
\begin{equation}
B_{s}'=\E_{s}\bbrack{\de_{s}\bpar{(X-b_{s})(X-b_{s})^{\T}}}-\cov_{s}\bpar{(X-b_{s})(X-b_{s})^{\T},q}=-\cov_{s}\bpar{(X-b_{s})(X-b_{s})^{\T},q}\,,\label{eq:unshifted-covariance-derivative}
\end{equation}
where the terms containing $\partial_{s}b_{s}$ vanish as $\E_{s}(X-b_{s})=0$.
Since $M_{s}=M^{1/2}B_{s}M^{1/2}$, \eqref{eq:unshifted-covariance-derivative}
implies
\begin{equation}
\partial_{s}M_{s}=-M^{1/2}\cov_{s}\bpar{(X-b_{s})(X-b_{s})^{\T},q}\,M^{1/2}\,.\label{eq:Ms-derivative}
\end{equation}

By perturbation theory for symmetric matrix paths \cite[Theorem 5.4 and Theorem 6.8]{Kato95perturbation},
the eigenvalues $\lambda_{i}(s)$ of $M_{s}$ may be labeled differentiably,
and at each $s$, one may choose an orthonormal eigenbasis $(u_{i})^{n}_{i=1}$
of $M_{s}$ such that $\partial_{s}\lambda_{i}(s)=u^{\T}_{i}(\partial_{s}M_{s})u_{i}$.
Hence, by \eqref{eq:Ms-derivative},
\begin{equation}
\partial_{s}\lambda_{i}(s)=-\cov_{s}\bbrace{\bpar{u^{\T}_{i}M^{1/2}(X-b_{s})}^{2},q}\,.\label{eq:deriv-lda}
\end{equation}
Using $M^{1/2}(X-b_{s})=M^{1/2}_{s}Z$ and $M^{1/2}_{s}u_{i}=\sqrt{\lambda_{i}(s)}\,u_{i}$,
\[
u^{\T}_{i}M^{1/2}(X-b_{s})=u^{\T}_{i}M^{1/2}_{s}Z=\sqrt{\lambda_{i}(s)}\,Z_{i}\,.
\]
Substituting this into \eqref{eq:deriv-lda} yields $\partial_{s}\lambda_{i}(s)=-\lambda_{i}(s)\cov_{s}(Z^{2}_{i},q)$,
which completes the proof.
\end{proof}

We use this identity to collect bounds on the key quantities (such
as $M_{s}$ and $y_{s}$) along the path indexed by $s$, from which
the theorem will follow. Recall that $\nu_{s}\propto\exp(-sx^{\T}Mx)\,\nu$
is a Gaussian tilt of an isotropic logconcave probability measure
$\nu$ in $\Rn$ for $s\geq0$, and that $b_{s}=\E_{s}X$ and $B_{s}=\cov_{s}X$
are its barycenter and covariance, respectively.
\begin{lem}
[Properties of $M_{s}$ and $y_{s}$]\label{lem:property-ms-ys}
The barycenter $y_{s}=M^{1/2}b_{s}$ and covariance $M_{s}=M^{1/2}B_{s}M^{1/2}$
satisfy the following bounds for $s\geq0$:
\begin{itemize}
\item $P_{s}+\norm{y_{s}}^{2}\le N_{s}\leq s\cdot\tr(M^{2}_{0})\,\msf Q_{n}$
\item $\tr(M^{2}_{s})\leq\tr(M^{2}_{0})\,(1+\msf Q_{n})$
\item $y^{\T}_{s}M_{s}y_{s}\leq\half\,\tr(M^{2}_{0})\,\msf Q_{n}$
\end{itemize}
\end{lem}

\begin{proof}
For every unit vector $\alpha\in\R^{n}$, by Cauchy--Schwarz followed
by the definition of $\cQ_{n}$ for the quadratic form $\sum_{i}\alpha_{i}Z^{2}_{i}$,
\[
\Babs{\sum_{i}\alpha_{i}\,\partial_{s}\log\lambda_{i}(s)}=\Babs{\cov_{s}\Bpar{\sum_{i}\alpha_{i}Z^{2}_{i},q}}\leq\sqrt{\cQ_{n}}\,\norm{\alpha}\sqrt{\var_{\nu_{s}}q}\,.
\]
Hence, $\sum_{i}[\partial_{s}\log\lambda_{i}(s)]^{2}\leq\cQ_{n}\var_{\nu_{s}}q$.
Substituting this into \eqref{eq:sum-log2},
\begin{equation}
\sum_{i}\log^{2}\frac{\lambda_{i}(s)}{\lambda_{i}(0)}\le s\int^{s}_{0}\sum_{i}\bpar{\partial_{r}\log\lambda_{i}(r)}^{2}\,\D r\leq s\cQ_{n}\int^{s}_{0}\var_{\nu_{r}}q\,\D r=s\cQ_{n}\Delta_{s}\,.\label{eq:unshifted-action}
\end{equation}
Using $\Delta_{s}\leq N_{s}$ and substituting this into \eqref{eq:unshifted-deficit},
$N^{2}_{s}\leq\tr(M^{2}_{0})\sum_{i}\log^{2}\frac{\lambda_{i}(s)}{\lambda_{i}(0)}\leq s\tr(M^{2}_{0})\,\cQ_{n}N_{s}$.
Hence,
\[
(P_{s}+\norm{y_{s}}^{2}\leq)N_{s}\leq s\cdot\tr(M^{2}_{0})\,\cQ_{n}\,.
\]
Thus, $P_{s},N_{s},\norm{y_{s}}^{2}=O(s)$ as desired. Substituting
this into \eqref{eq:final-bound}, we obtain $\tr(M^{2}_{s})\leq\tr(M^{2}_{0})\,(1+\cQ_{n})$
and $y^{\T}_{s}M_{s}y_{s}\leq\frac{1}{2}\,\tr(M^{2}_{0})\,\cQ_{n}$.
\end{proof}

Substituting the second and third bounds into \eqref{eq:unshifted-target},
we obtain, for every $s>0$,
\[
\var_{s}q\le(2\cQ^{2}_{n}+6\cQ_{n})\,\tr(M^{2})\,,
\]
which completes the proof of Theorem~\ref{thm:centered-anisotropic}.
When $M\succeq0$, we apply the result to $M+\veps I_{n}$ and let
$\veps\downarrow0$.
\begin{rem}
[Collective control of eigenvalues] The proof controls the eigenvalues
of $M_{s}$ \emph{collectively}, rather than individually. For example,
considering the total continuous change in the trace (specifically,
$N_{s}$ and $P_{s}$ terms), we obtain a useful bound. If we attempt
to do this for each eigenvalue, then we may obtain some bound on each,
but summing them up in a na\"ive way would lead to a factor of $n$.
A further consequence is the second item of Lemma~\ref{lem:property-ms-ys}
(i.e., the stability of the Schatten $2$-norm along the Gaussian
tilt):
\[
\tr(M^{2}_{s})\leq\tr(M^{2}_{0})\,(1+\msf Q_{n})\,.
\]
As noted earlier, this type of bound does not hold for individual
eigenvalues such as the largest eigenvalue (i.e., the operator norm).
\end{rem}

\subsection{Stability under shifted Gaussian tilts\label{sec:shifted-tilts}}

We now consider the shifted quadratic: $q_{u}(x)=\norm{M^{1/2}x+u}^{2}$.
In this case, the initial transformed mean is $y_{0}=u$, and we retain
the same variance decomposition:
\begin{equation}
\sqrt{\var_{s}q}\leq\sqrt{\cQ_{n}\tr(M^{2}_{s})}+2\sqrt{y^{\T}_{s}M_{s}y_{s}}\,.\label{eq:var-decomp}
\end{equation}
The identity becomes $\Delta_{s}=N_{s}-P_{s}+\norm u^{2}-\norm{y_{s}}^{2}$,
so $P_{s}+\norm{y_{s}}^{2}\leq N_{s}$ is replaced by the weaker bound
$P_{s},\norm{y_{s}}^{2}\leq\norm u^{2}+N_{s}$. Applying Brascamp--Lieb
again, this extra $\norm u^{2}$ term yields $\norm u^{2}/s$. The
resulting bound is useful for moderate and large $s$ but vacuous
at the untilted endpoint (i.e., $s=0$).

The proof follows the same steps and separately treats the regime
near $s=0$, solving the Gr\"onwall-type inequality $\abs{v_{s}'}\lesssim v^{3/2}_{s}$.
\begin{proof}
[Proof of Theorem~\ref{thm:anisotropic}] Assume $M\succ0$ and
abbreviate $q=q_{u}$. Define $v_{s}:=\var_{s}q$,  $G:=\tr(M^{2})+u^{\T}Mu$,
and
\[
b_{s}:=\E_{s}X\,,\qquad B_{s}:=\cov_{s}X\,,\qquad M_{s}:=M^{1/2}B_{s}M^{1/2}\,,\qquad y_{s}:=M^{1/2}b_{s}+u\,.
\]
By isotropy, we have $B_{0}=\Id$, $M_{0}=M$, and $y_{0}=u$.

By Lemma~\ref{lem:score-identity}, we similarly have
\[
(\E_{s}q)'=-v_{s}\,,\qquad\Delta_{s}:=\E_{0}q-\E_{s}q=\int^{s}_{0}v_{r}\,\D r\,,\qquad B_{s}'=-\cov_{s}\bpar{(X-b_{s})(X-b_{s})^{\T},q}\,.
\]
Let $\lambda_{i}(s)$ be an absolutely continuous labeling of the
eigenvalues of $M_{s}$. The calculation in the preceding subsection
gives the same identity: $\de_{s}\log\lda_{i}(s)=-\cov_{s}(Z^{2}_{i},q)$.
Hence,
\[
\sum_{i}\bpar{\de_{s}\log\lda_{i}(s)}^{2}\leq\cQ_{n}v_{s}\qquad\text{and}\qquad\sum_{i}\log^{2}\frac{\lambda_{i}(s)}{\lambda_{i}(0)}\leq\cQ_{n}s\Delta_{s}\,.
\]

Define $N_{s}$ and $P_{s}$ as before. Since $\E_{s}q=\tr M_{s}+\norm{y_{s}}^{2}$,
\[
\Delta_{s}=N_{s}-P_{s}+\norm u^{2}-\norm{y_{s}}^{2}\qquad\text{and}\qquad\Delta_{s},P_{s},\norm{y_{s}}^{2}\leq\norm u^{2}+N_{s}\,.
\]
Then, $N_{s}\leq\sqrt{\cQ_{n}\tr(M^{2})\,s\Delta_{s}}\leq\sqrt{\cQ_{n}\tr(M^{2})\,s\,(\norm u^{2}+N_{s})}$.
Solving this,
\[
N_{s}\le\cQ_{n}\tr(M^{2})\,s+\norm u\sqrt{\cQ_{n}\tr(M^{2})\,s}\,.
\]
Combining this with the bounds on $P_{s}$ and $y_{s}$ gives
\[
\tr(M^{2}_{s})\lesssim(1+\cQ_{n})\tr(M^{2})+\frac{\norm u^{2}}{s}\qquad\text{and}\qquad y^{\T}_{s}M_{s}y_{s}\lesssim\cQ_{n}\tr(M^{2})+\frac{\norm u^{2}}{s}\,.
\]
Substituting these bounds into \eqref{eq:var-decomp},
\begin{equation}
v_{s}\lesssim(\cQ_{n}+\cQ^{2}_{n})\tr(M^{2})+(1+\cQ_{n})\,\frac{\norm u^{2}}{s}\,.\label{eq:general-bound}
\end{equation}
The last term is the only singular term introduced by the shift, and
it becomes vacuous as $s$ goes to $0$. However, at $s=0$, the same
decomposition gives $v_{0}\leq(\cQ_{n}+4)\,G$. To combine the two
regimes, we propagate the endpoint bound to small $s$.

\paragraph{Handling small $s$.}

By Lemma~\ref{lem:score-identity}, $v_{s}'=-\E_{s}[(q-\E_{s}q)^{3}]$.
By the polynomial reverse H\"older inequality \cite[Theorem 7]{CW01distributional},
$\abs{v_{s}'}\lesssim v^{3/2}_{s}$. Let $\bar{v}:=(\cQ_{n}+4)\,G$
and $s_{0}=c\bar{v}^{-1/2}$ for sufficiently small $c>0$. Since
$M\succ0$ and $\nu_{s}$ is full-dimensional, $v_{s}>0$. Integrating
the differential inequality, for $s\in[0,s_{0}]$, we have $v_{s}\leq4\bar{v}$.
Since $u^{\T}Mu\leq\norm u^{2}\norm M_{\op}\leq\norm u^{2}\sqrt{G}$,
we have $G=\tr(M^{2})+u^{\T}Mu\leq\tr(M^{2})+\norm u^{2}\sqrt{G}$.
Hence, $v_{s}\lesssim(1+\msf Q_{n})\,(\tr(M^{2})+\norm u^{2}\sqrt{G})$. 

For $s\geq s_{0}$, the pointwise bound \eqref{eq:general-bound}
yields
\[
v_{s}\lesssim(\cQ_{n}+\cQ^{2}_{n})\tr(M^{2})+(1+\cQ_{n})\,\norm u^{2}\sqrt{(\cQ_{n}+4)\,G}\,.
\]
Therefore,
\[
v_{s}\lesssim(1+\msf Q_{n})^{2}\,\bpar{\tr(M^{2})+\norm u^{2}\sqrt{G}}\,.
\]
This proves the claim when $M\succ0$. For $M\succeq0$, apply the
result to $M+\varepsilon I_{n}$ and let $\varepsilon\downarrow0$;
dominated convergence follows from the finite fourth moment of $\nu$.
\end{proof}

We now transfer the shifted stability bound to the original coordinates.
\begin{proof}
[Proof of Theorem~\ref{thm:main-stability}] Since $\pi$ is full-dimensional,
$\Sigma\succ0$. Let $Z=\Sigma^{-1/2}(X-m)$ for $X\sim\pi$, and
let $\nu$ be the law of $Z$, which is isotropic and logconcave.
Under the same change of variables, the tilted law $\pi_{t}$ becomes
\[
\nu_{s}(\D z)\propto\exp(-s\,\norm{\Sigma^{1/2}z+m}^{2})\,\nu(\D z)\qquad\text{for }s=(2t)^{-1}\,,
\]
and $\norm X^{2}=\norm{\Sigma^{1/2}Z+m}^{2}$. Applying Theorem~\ref{thm:anisotropic}
with $M=\Sigma$ and $u=m$ proves that for $t>0$,
\[
\var_{\pi_{t}}(\norm X^{2})\lesssim(1+\msf Q_{n})^{2}\,\bpar{\tr(\Sigma^{2})+\norm m^{2}\sqrt{\tr(\Sigma^{2})+m^{\T}\Sigma m}}\,.
\]
To obtain the simpler scales, note that 
\[
\tr(\Sigma^{2})\leq\Lambda\tr\Sigma\leq R^{2}\Lambda\,,\qquad\tr(\Sigma^{2})+m^{\T}\Sigma m\leq\Lambda\,(\tr\Sigma+\norm m^{2})=R^{2}\Lambda\,,\qquad\norm m^{2}\leq L\,.
\]
Consequently,
\[
\tr(\Sigma^{2})+\norm m^{2}\sqrt{\tr(\Sigma^{2})+m^{\T}\Sigma m}\leq R^{2}\Lambda+LR\sqrt{\Lambda}\leq2\min\{R^{2}L,R^{3}\sqrt{\Lambda}\}\,.
\]
Here, for the $R^{2}L$ bound, the last step uses $\Lambda\leq L$
and $\sqrt{\Lambda}\leq R$; for the $R^{3}\sqrt{\Lambda}$ bound,
it uses $\sqrt{\Lambda}\leq R$ and $L\leq R^{2}$.
\end{proof}

\section{Faster annealing for uniform sampling\label{sec:uniform}}

In \S\ref{subsec:renyi-close-unif}, we relate variance stability
to R\'enyi divergence between consecutive Gaussian annealing distributions.
This determines the step size in the Gaussian variance $\sigma^{2}$
in our annealing scheme. Then, in \S\ref{subsec:unif-gc}, we give
a faster version of Gaussian cooling, proving Theorem~\ref{thm:warm-uniform}.
To go through the main calculations, the reader might finding it convenient
to think of $q$ as a fixed constant, say $2$. 

\subsection{R\'enyi closeness of annealing distributions\label{subsec:renyi-close-unif}}

For a probability measure $\nu$, let $\nu_{s}(\D x)\propto e^{-s\,\norm x^{2}}\,\nu(\D x)$
for $s\ge0$, and $F(s):=\log\int e^{-s\,\norm x^{2}}\,\nu(\D x)$.
Differentiating gives $F'(s)=-\E_{\nu_{s}}[\norm X^{2}]$ and $F''(s)=\var_{\nu_{s}}(\norm X^{2})$.
Recall $q$-R\'enyi divergence: $\msf R_{q}(\mu\mmid\nu):=\frac{1}{q-1}\log\int(\frac{\D\mu}{\D\nu})^{q}\,\D\nu$
for $q>1$.
\begin{lem}
\label{lem:renyi-exact} Let $q>1$ and $s\ge s'\ge0$. Then, 
\begin{equation}
\msf R_{q}(\nu_{s}\mmid\nu_{s'})\le\frac{q}{2}\,(s-s')^{2}\sup_{t\in[s',qs-(q-1)\,s']}\var_{\nu_{t}}(\norm X^{2})\,.\label{eq:renyi-curvature}
\end{equation}
\end{lem}

\begin{proof}
Let $Z(s):=\int e^{-s\,\norm x^{2}}\,\D\nu(x)$. For $\Delta:=s-s'$,
the density ratio is $\frac{\D\nu_{s}}{\D\nu_{s'}}(x)=\frac{Z(s')}{Z(s)}\,e^{-\Delta\,\norm x^{2}}$,
and
\[
\int\bpar{\frac{\D\nu_{s}}{\D\nu_{s'}}}^{q}\,\D\nu_{s'}=\frac{Z(s')^{q-1}\,Z\bpar{q\,s-(q-1)\,s'}}{Z(s)^{q}}\,.
\]
Thus,
\[
\msf R_{q}(\nu_{s}\mmid\nu_{s'})=\frac{F\bpar{s+(q-1)\,\Delta}-qF(s)+(q-1)\,F(s-\Delta)}{q-1}\,.
\]
Let $I:=[s',s+(q-1)\Delta]$ and $V_{I}:=\sup_{t\in I}\var_{\nu_{t}}(\norm X^{2})=\sup_{t\in I}F''(t)$.
By Taylor expansion,
\begin{align*}
F\bpar{s+(q-1)\,\Delta} & \le F(s)+(q-1)\,\Delta F'(s)+\frac{V_{I}}{2}\,(q-1)^{2}\Delta^{2}\,,\\
F(s-\Delta) & \le F(s)-\Delta F'(s)+\frac{V_{I}}{2}\,\Delta^{2}\,.
\end{align*}
Substituting these estimates into $\msf R_{q}$,
\begin{align*}
\msf R_{q}(\nu_{s}\mmid\nu_{s'}) & \leq\frac{F(s)+(q-1)\,\Delta F'(s)+\frac{V_{I}}{2}\,(q-1)^{2}\Delta^{2}-qF(s)+(q-1)\,\bpar{F(s)-\Delta F'(s)+\frac{V_{I}}{2}\,\Delta^{2}}}{q-1}\\
 & =\frac{qV_{I}}{2}\,\Delta^{2}\,,
\end{align*}
which completes the proof.
\end{proof}

This result does not require logconcavity of $\nu$. For $\sigma^{2}>0$,
define $\mu_{\sigma^{2}}:=\nu_{1/(2\sigma^{2})}$. The preceding lemma
gives the following bounds; apply Lemma~\ref{lem:renyi-exact} with
$s=(2\sigma^{2})^{-1}$ and $s'=(2\sigma^{2}\,(1+\alpha))^{-1}$,
and take $s'=0$ for the second bound.
\begin{cor}
[Closeness]\label{cor:renyi-t} Let $\cV:=\sup_{s\ge0}\var_{\nu_{s}}(\norm X^{2})$.
For every $q>1$, $\sigma^{2}>0$, and $\alpha>0$,
\[
\msf R_{q}(\mu_{\sigma^{2}}\mmid\mu_{\sigma^{2}\,(1+\alpha)})\le\frac{q\cV}{8\sigma^{4}}\,\bpar{\frac{\alpha}{1+\alpha}}^{2}\le\frac{q\cV\alpha^{2}}{8\sigma^{4}}\qquad\text{and}\qquad\msf R_{q}(\mu_{\sigma^{2}}\mmid\nu)\le\frac{q\cV}{8\sigma^{4}}\,.
\]
In particular, the update $\sigma^{2}\mapsto\sigma^{2}\,(1+\alpha)$
with $\alpha\lesssim\sigma^{2}/\sqrt{q\cV}$ maintains $O(1)$-closeness,
and $\sigma^{2}\gtrsim\sqrt{q\cV}$ ensures terminal $O(1)$-warmness.
\end{cor}

\begin{rem}
[Comparison with prior closeness results]\label{rem:comparison-closeness}
\cite[Lemma 4.5]{KV25faster} bounds $\msf R_{q}(\mu_{\sigma^{2}}\mmid\mu_{\sigma^{2}\,(1+\alpha)})$
by $qD^{2}\alpha^{2}/\sigma^{2}$ for compactly supported logconcave
$\nu$ with support diameter $D$. Later, \cite[Lemma 4.2]{KV26zeroLC}
gives the bound $q^{2}R^{2}\alpha^{2}/(2\sigma^{2})$ for the same
quantity when $\E_{\nu}[\norm X^{2}]=R^{2}\leq D^{2}$.

Our estimate~\eqref{eq:renyi-curvature} yields an improved bound.
For every $t>0$, the measure $\nu_{t}$ is $2t$-strongly logconcave.
By Brascamp--Lieb or \eqref{eq:pi} with $\cpi(\nu_{t})\leq1/(2t)$,
\[
\var_{\nu_{t}}(\norm X^{2})\leq\frac{1}{2t}\,\E_{\nu_{t}}[4\,\norm X^{2}]=\frac{2}{t}\,\E_{\nu_{t}}[\norm X^{2}]\leq\frac{2R^{2}}{t}\,,
\]
where the last inequality follows from $\E_{\nu_{t}}[\norm X^{2}]\leq\E_{\nu}[\norm X^{2}]$.
For $s=(2\sigma^{2})^{-1}$ and $s'=[2\sigma^{2}\,(1+\alpha)]^{-1}$
with $t\geq s'$, we have
\[
\msf R_{q}(\mu_{\sigma^{2}}\mmid\mu_{\sigma^{2}(1+\alpha)})\leq\frac{q}{2}\,\bpar{\frac{\alpha}{2\sigma^{2}\,(1+\alpha)}}^{2}\,\frac{2R^{2}}{s'}=\frac{qR^{2}\alpha^{2}}{2\sigma^{2}\,(1+\alpha)}\,.
\]
This improves the $q^{2}$-dependence of the second bound to $q$.
It also recovers the compact-support bound in the present convex-body
setting (due to $R^{2}\leq D^{2}$).
\end{rem}

\subsection{Gaussian cooling through thin-shell stability\label{subsec:unif-gc}}

Let $\K\subset\R^{n}$ be a convex body with $B(0,1)\subset\K$, and
let $\pi$ be uniform on $\K$. Fix a R\'enyi order $q$ and take
$\cV=\Otilde(R^{2}L\wedge R^{3}\sqrt{\Lambda})$ by Theorem~\ref{thm:main-stability}.
We propose the following annealing schedule:
\begin{itemize}
\item Let $\sigma^{2}_{\textup{start}}:=1/n$ and $\sigma^{2}_{\textup{last}}:=\sqrt{q\cV}$
be the initial and terminal variance parameters.
\item At phase $i$, the annealing distribution is $\mu_{\sigma^{2}_{i}}\propto\exp\{-\norm x^{2}/(2\sigma^{2}_{i})\}\cdot\ind_{\K}$.
For $\sigma^{2}_{i}\in[\sigma^{2}_{\textup{start}},\sigma^{2}_{\textup{last}}]$,
define $\sigma^{2}_{i+1}$ recursively by
\begin{equation}
\sigma^{2}_{i+1}:=\sigma^{2}_{i}\,\bpar{1+\frac{\sigma^{2}_{i}}{\sqrt{q\cV}}}\,.\label{eq:uniform-schedule}
\end{equation}
\end{itemize}
By Corollary~\ref{cor:renyi-t}, consecutive distributions are $1/4$-close
at $\msf R_{2q}$.

Next, we recall a sampling algorithm for truncated Gaussian distributions
\cite{KV26zeroLC}. Starting from $x_{0}\sim\mu_{0}$, with step size
$h>0$ and threshold $\tau$, set $s:=\sigma^{2}/(h+\sigma^{2})\in(0,1)$.
One iteration of $\psgauss$ consists of the following two steps:
\begin{itemize}
\item {[}Forward{]} Draw $y_{k+1}\sim\msf N(x_{k},hI_{n})$.
\item {[}Backward{]} Repeatedly draw $x_{k+1}\sim\msf N(sy_{k+1},shI_{n})$
until $x_{k+1}\in\K$. If the number of rejection trials exceeds $\tau$,
declare failure.
\end{itemize}
This sampler follows the annealing path efficiently because the previous
annealing distribution is an $O(1)$-warm start for the next target
distribution.
\begin{prop}
[Truncated-Gaussian sampler]\label{prop:truncated-gaussian-relay}
Let $\mu_{\sigma^{2}}\propto e^{-\norm x^{2}/(2\sigma^{2})}\,\ind_{\K}(x)$
for a convex body $\K\supset B(0,1)$ given by a membership oracle.
Given $\varepsilon>0$, $\eta\in(0,1/2)$, $q\ge q_{0}\ge2$, and
an initial distribution $\mu_{0}$ with $M_{q_{0}}:=\norm{\nicefrac{\D\mu_{0}}{\D\mu_{\sigma^{2}}}}_{L^{q_{0}}(\mu_{\sigma^{2}})}\leq10$,
initialize $\psgauss$ with $\mu_{0}$ and iterate it $N=\Otilde(n^{2}\sigma^{2}\log\frac{q}{\veps}\log\frac{1}{\eta})$
times with $h=(10n^{2}\log Z)^{-1}$ and $\tau=Z^{2}\log^{3}Z$, where
$Z=16NM_{2}/\eta$. Assume $M_{q_{0}}\leq10$ and $q_{0}\geq1+\log\tau$.
\begin{itemize}
\item The latter condition is ensured by $q_{0}\ge2\vee O(\log\frac{N}{\eta})$.
\item With probability at least $1-\eta$, $\psgauss$ completes all $N$
iterations without failure. Conditioned on this event, its output
has law $\nu$ satisfying $\msf R_{q}(\nu\mmid\mu_{\sigma^{2}})\le\varepsilon+2\log\frac{1}{1-\eta}$,
using $\Otilde(n^{2}\sigma^{2}\log\frac{q}{\veps}\log^{5}\frac{1}{\eta})$
membership queries in expectation. The corresponding uncapped chain
(i.e., $\tau=\infty$) has output law $\bar{\nu}$ satisfying $\msf R_{q}(\bar{\nu}\mmid\mu_{\sigma^{2}})\le\varepsilon$.
\end{itemize}
\end{prop}

Thus, we need $\Otilde(n^{2}\sigma^{2}_{i}\log q)$ membership queries
at stage $i$. Let $J$ be the first stage with $\sigma^{2}_{J}\ge\sigma^{2}_{\textup{last}}$.

\begin{lyxalgorithm}
[Thin-shell Gaussian cooling]\label{alg:uniform-annealing} \textbf{\emph{Input}}:
membership oracle for $\K$ with $B(0,1)\subset\K$, R\'enyi-order
$q$, and $\cV=\Otilde(R^{2}L\wedge R^{3}\Lambda^{1/2})$. 
\begin{enumerate}
\item Initialize $i=0$, and draw $X_{0}\sim\mu_{0}:=\mu_{\sigma^{2}_{\textup{start}}}$
by rejection sampling from $\msf N(0,\sigma^{2}_{\textup{start}}\,I_{n})$.
\item While $\sigma^{2}_{i}<\sigma^{2}_{\textup{last}}$, 
\begin{enumerate}
\item compute $\sigma^{2}_{i+1}$ by \eqref{eq:uniform-schedule}: $\sigma^{2}_{i+1}:=\sigma^{2}_{i}\,\bpar{1+\frac{\sigma^{2}_{i}}{\sqrt{q\cV}}}\,.$
\item Iterate $\psgauss:\msf R_{q}\to\msf R_{q}$ for $N_{i+1}=\Otilde(n^{2}\sigma^{2}_{i+1}\log q\log J)$
iterations with target $\mu_{\sigma^{2}_{i+1}}$, accuracy $\veps_{i+1}=1/20$,
and failure $\eta_{i+1}=1/(100J)$. Use the parameters of the Truncated
Gaussian Sampler (Proposition~\ref{prop:truncated-gaussian-relay})
with $Z_{i+1}:=160N_{i+1}/\eta_{i+1}$, restarting if a cap is hit.
After each successful phase, replace $i$ by $i+1$, and return the
current point if $\sigma^{2}_{i}\ge\sigma^{2}_{\textup{last}}$.
\end{enumerate}
\end{enumerate}
\end{lyxalgorithm}

\begin{proof}
[Proof of Theorem~\ref{thm:warm-uniform}] Let $q_{\min}:=2\vee\Theta(\log(en^{3}\cV))$.
For $q<q_{\min}$, run the algorithm at order $q_{\min}$ and use
monotonicity of R\'enyi divergence; this increases the cost by only
a polylogarithmic factor. Thus, we may assume $q\geq q_{\min}$. Define
$t_{i}:=\sigma^{2}_{i}$. The initial rejection sampler produces $X_{0}\sim\mu_{t_{0}}$
exactly and takes $\Theta(1)$ expected trials, since $t_{0}=n^{-1}$
and $B(0,1)\subset\K$. As in \cite[High-level idea of analysis]{KV26zeroLC},
we first consider the \emph{ideal} truncated-Gaussian sampler where
every threshold is removed (i.e., $\tau=\infty$) and then bound the
bias introduced by the cap. Let $\bar{\nu}_{i}$ be its output law
after phase $i$ (so $\bar{\nu}_{0}=\mu_{t_{0}}$).

The uncapped mixing estimate in Proposition~\ref{prop:truncated-gaussian-relay},
applied with input order $q$, output order $2q$, and accuracy $1/20$,
ensures $\msf R_{2q}(\bar{\nu}_{i}\mmid\mu_{t_{i}})\le\frac{1}{20}$
for $0\leq i\leq J$. For $i<J$, the update is $\alpha_{i}=t_{i}/\sqrt{q\msf V}$,
so Corollary~\ref{cor:renyi-t} gives $\msf R_{2q}(\mu_{t_{i}}\mmid\mu_{t_{i+1}})\le1/4$.
By the weak triangle inequality,
\[
\msf R_{q}(\bar{\nu}_{i}\mmid\mu_{t_{i+1}})\le\frac{q-1/2}{q-1}\,\msf R_{2q}(\bar{\nu}_{i}\mmid\mu_{t_{i}})+\msf R_{2q-1}(\mu_{t_{i}}\mmid\mu_{t_{i+1}})\le\frac{13}{40}\,.
\]
Hence, the warmness satisfies $M_{q}<e$. The schedule gives $J=O(n\sqrt{q\cV})$
and $\sigma^{2}_{i}\leq2\sqrt{q\cV}$, so $\log\tau_{i}\lesssim\log(en^{3}q(1\vee\cV))$.
Since $q\geq q_{\min}$, choosing the hidden constant in $q_{\min}$
sufficiently large ensures $q\geq1+\log\tau_{i}$ at every phase,
as required by Proposition~\ref{prop:truncated-gaussian-relay}.

We now return to the capped algorithm with threshold $\tau$. Couple
a capped warm-start algorithm to the ideal warm-start algorithm, using
the same randomness until a cap is hit, and let $S$ be the event
that all $J$ phases succeed. The per-phase failure estimate in Proposition~\ref{prop:truncated-gaussian-relay}
is taken under the ideal input law. As in \cite[High-level idea of analysis]{KV26zeroLC},
the union bound gives $\P(S)\ge1-\sum^{J}_{i=1}\eta_{i}\ge99/100$.
On $S$, the capped and ideal paths coincide, so restarting independent
full attempts until success returns $\nu=\law(\bar{X}_{J}\mid S)$.
At the terminal index, $\msf R_{2q}(\mu_{t_{J}}\mmid\pi)\le1/4$ by
Corollary~\ref{cor:renyi-t}, and applying the weak triangle inequality
gives $\msf R_{q}(\bar{\nu}_{J}\mmid\pi)\le13/40$. Finally, $\frac{\D\nu}{\D\bar{\nu}_{J}}(x)=\frac{\P(S\mid\bar{X}_{J}=x)}{\P(S)}\le\frac{1}{\P(S)}$,
so
\[
\msf R_{q}(\nu\mmid\pi)\le\msf R_{q}(\bar{\nu}_{J}\mmid\pi)+\frac{q}{q-1}\log\frac{1}{\P(S)}<1\,.
\]

It remains to bound the query complexity. Starting from a given $\sigma^{2}_{i}$,
doubling this value requires $O(\sqrt{q\cV}/\sigma^{2}_{i})$ phases.
Since the query complexity in each phase during doubling is $\Otilde(n^{2}\sigma^{2}_{i}\log q\log^{5}J)$,
the total query cost of a doubling is $\Otilde(q^{1/2}n^{2}\msf V^{1/2})$.
Since there are only logarithmically many doublings, the total query
complexity is 
\[
\Otilde(q^{1/2}n^{2}\msf V^{1/2})\times O\bpar{\log_{2}(n\sqrt{q\msf V})}=\Otilde(q^{1/2}n^{2}\msf V^{1/2})\,.
\]
If $C$ denotes the cost of one full attempt, then the restart wrapper
has expected cost $\E C_{\mathrm{restart}}=\frac{\E C}{\P(S)}\le\frac{100}{99}\,\E C$.
Thus, restarting increases the cost by only a universal factor.
\end{proof}

\section{Extension to logconcave sampling\label{sec:logconcave}}

Let $V:\R^{n}\to\R\cup\{+\infty\}$ be a convex function such that
$\D\pi(x)\propto e^{-V(x)}\,\D x$ is a full-dimensional logconcave
probability measure in $\Rn$. Assume that $V$ is given by an evaluation
oracle, which returns the value of $V$ at the queried point. Denote
$v_{*}:=\inf V$, but neither $v_{*}$ nor a minimizer of $V$ is
assumed known. The ground set is $\msf L_{g}:=\{x\in\R^{n}:V(x)-\inf V\le10n\}$,
and we impose the normalization $B(0,1)\subset\msf L_{g}$ following
\cite{KV25sampling}.

We first recall the preprocessing steps from \cite{KV25sampling}.

\paragraph{Reduction.}

Define the epigraph\footnote{We retain the original $t$-coordinate and do not normalize $\inf V$.}
\begin{equation}
\K:=\{(x,t)\in\R^{n}\times\R:V(x)\le nt\}\,.\label{eq:shifted-lift}
\end{equation}
Let $a:=V(0)/n$, which is known from one oracle query. Since $0\le V(0)-v_{*}\le10n$,
every point of $\K$ has $t\ge a-10$. Moreover, $V(x)\le V(0)+10n$
for $x\in B(0,1)$, so $B_{n+1}((0,a+11),1)\subset\K$.

Instead of sampling directly from $\pi$, we use the exponential-lifting
technique \cite{KV25sampling} and define the lifted distribution
by
\[
\pi^{X,T}(\D x,\D t)\propto e^{-nt}\,\ind_{\K}(x,t)\,\D x\D t\,,
\]
whose $X$-marginal is the original target $\pi$, since integrating
over $t$ gives a density proportional to $e^{-V(x)}$.

\paragraph{Constant-mass truncation.}

To ensure that the lifted distribution has a finite log-Sobolev constant,
let $\ell:=\log(4e)$, $D:=1\vee R\ell$, and $b:=13\ell+5$. We then
truncate the epigraph $\K$ as follows:
\[
\bar{\K}:=\K\cap\bpar{B(0,D)\times[a-10,a+b]}\,.
\]
Since $D\ge1$ and $b>12$, the ball $B_{n+1}((0,a+11),1)$ remains
inside $\bar{\K}$. We now verify that this truncation has constant
mass. If $E\sim\Exp(1)$ is independent of $X\sim\pi$, then the conditional
law of the lift gives
\[
T=\frac{V(X)+E}{n}\,.
\]
By \cite[Theorem 2.3 and Corollary 4.5]{FMW16varentropy}, $\E_{\pi}[V(X)-v_{*}]\le n$
and $\var_{\pi}V\le n$. Hence,
\[
\E T\le a+2\qquad\text{and}\qquad\var T=\frac{\var_{\pi}V+1}{n^{2}}\le2\,.
\]
By Markov's and Chebyshev's inequalities,
\[
\pi^{X,T}(\bar{\K}^{c})\le\P\bpar{\norm X>R\ell}+\P(T>a+b)\le\frac{1}{\ell^{2}}+\frac{2}{(b-2)^{2}}<\frac{1}{2}\,,
\]
so $p:=\pi^{X,T}(\bar{\K})\ge1/2$. Let $\bar{\pi}^{X,T}\propto\pi^{X,T}\,\ind_{\bar{\K}}$
be the truncated distribution, and write $\bar{\pi}$ for its $X$-marginal.
Clearly, $\D\bar{\pi}/\D\pi\le p^{-1}\le2$. Note that the truncated
lift $\bar{\pi}^{X,T}$ is logconcave.

The truncated $X$-marginal need not have the same mean as $\pi$.
The following comparison shows that the relevant parameters change
by only a constant factor.
\begin{lem}
\label{lem:subdensity} Let $\pi$ and $\bar{\pi}$ be probability
measures on $\Rn$ such that $\bar{\pi}$ is logconcave, and $\D\bar{\pi}/\D\pi\le p^{-1}_{0}$
for some $p_{0}>0$. Let $\bar{R},\bar{L},\bar{\Lambda}$ denote the
analogues of $R,L,\Lambda$ for $\bar{\pi}$. Then:
\begin{itemize}
\item $\E_{\bar{\pi}}[XX^{\T}]\preceq p^{-1}_{0}\E_{\pi}[XX^{\T}]$ and
$\cov_{\bar{\pi}}X\preceq p^{-1}_{0}\Sigma$;
\item $\bar{R}^{2}\le p^{-1}_{0}R^{2}$, $\bar{L}\le p^{-1}_{0}L$, and
$\bar{\Lambda}\le p^{-1}_{0}\Lambda$;
\item $\sup_{\sigma^{2}>0}\var_{\bar{\pi}_{\sigma^{2}}}(\norm X^{2})\lesssim p^{-2}_{0}(1+\msf Q_{n})^{2}\,(R^{2}L\wedge R^{3}\sqrt{\Lambda})$.
\end{itemize}
\end{lem}

\begin{proof}
For every $u\in\R^{n}$, we have $\E_{\bar{\pi}}[\inner{u,X}^{2}]\leq p^{-1}_{0}\E_{\pi}[\inner{u,X}^{2}]$.
For $m=\E_{\pi}X$, 
\[
\var_{\bar{\pi}}(\inner{u,X})\le\E_{\bar{\pi}}[\langle u,X-m\rangle^{2}]\le p^{-1}_{0}\E_{\pi}[\langle u,X-m\rangle^{2}]=p^{-1}_{0}\var_{\pi}(\inner{u,X})\,,
\]
which proves the first item. The second item follows by taking traces
and operator norms. Applying Theorem~\ref{thm:main-stability} to
$\bar{\pi}$ and using $\bar{R}^{2}\bar{L}\le p^{-2}_{0}R^{2}L$ and
$\bar{R}^{3}\sqrt{\bar{\Lambda}}\le p^{-2}_{0}R^{3}\sqrt{\Lambda}$
proves the third item.
\end{proof}

\subsection{R\'enyi closeness of lifted annealing distributions}

As in \cite{KV25faster,KV26zeroLC}, we use the following annealing
distribution with parameters $\rho\ge0$ and $\sigma^{2}>0$: 
\begin{equation}
\mu_{\sigma^{2},\rho}(\D x,\D t)\propto\exp\bpar{-\frac{\norm x^{2}}{2\sigma^{2}}-\rho t}\,\ind_{\bar{\K}}(x,t)\,\D x\D t\,,\label{eq:two-parameter-path}
\end{equation}
and follow a two-phase annealing schedule, which simplifies the schedules
in \cite{KV25faster,KV26zeroLC}. In the $(\sigma^{2},\rho)$-plane,
we first follow $(n^{-1},1)\to(n^{-1},n)$ and then increase $\sigma^{2}$
with $\rho=n$ fixed. As in \S\ref{sec:uniform}, we bound the R\'enyi
divergence between neighboring targets, follow the annealing path
with a logconcave sampler, and sum the per-phase query costs. 

\paragraph{R\'enyi closeness.}

In the first phase $(n^{-1},1)\to(n^{-1},n)$, we use the multiplicative
update $\rho\gets\rho\,(1+O(\frac{1}{[q\,(q+n)]^{1/2}}))$. To justify
this update, we bound the divergence incurred by changing the $t$-tilt:
\begin{lem}
\label{lem:t-closeness} Let $\nu_{s}:=\mu_{n^{-1},s}$ for $s>0$.
Then, for $q>1$ and $0\le\alpha<(q-1)^{-1}$ such that $(1+\alpha)\,s\le n$,
\[
\msf R_{q}(\nu_{s}\mmid\nu_{(1+\alpha)\,s})\lesssim\frac{qn\alpha^{2}}{\bpar{1-(q-1)\,\alpha}^{2}}\,.
\]
\end{lem}

\begin{proof}
Let $\Phi(s):=\log\int_{\bar{\K}}\exp(-n\,\norm x^{2}/2-st)\,\D x\D t$.
As in the proof of Lemma~\ref{lem:renyi-exact}, we obtain
\[
\msf R_{q}(\nu_{s}\mmid\nu_{(1+\alpha)s})=\frac{\Phi\bpar{s\,(1-(q-1)\,\alpha)}-q\Phi(s)+(q-1)\,\Phi\bpar{(1+\alpha)\,s}}{q-1}\,.
\]
Set $s_{-}:=s\,(1-(q-1)\,\alpha)>0$, $s_{+}:=s\,(1+\alpha)$, $I:=[s_{-},s_{+}]$,
and $V_{I}:=\sup_{u\in I}\Phi''(u)=\sup_{u\in I}\var_{\nu_{u}}T$.
By Taylor expansion,
\begin{align*}
\Phi(s_{-}) & \leq\Phi(s)-(q-1)\,\alpha s\Phi'(s)+\frac{(q-1)^{2}\alpha^{2}s^{2}}{2}\,V_{I}\,,\\
\Phi(s_{+}) & \leq\Phi(s)+\alpha s\Phi'(s)+\frac{\alpha^{2}s^{2}}{2}\,V_{I}\,.
\end{align*}
Substituting these estimates into the $\msf R_{q}$-identity yields
\[
\msf R_{q}(\nu_{s}\mmid\nu_{(1+\alpha)s})\leq\frac{(q-1)^{2}\alpha^{2}s^{2}V_{I}+(q-1)\,\alpha^{2}s^{2}V_{I}}{2\,(q-1)}=\frac{q\alpha^{2}s^{2}}{2}\,V_{I}\,.
\]

We now bound $V_{I}$. To this end, we first bound the second moment
of $\nu_{u}$ as $\E_{\nu_{u}}[\norm X^{2}]\lesssim1$ uniformly over
$u\in[0,n]$. Note that the $X$-marginal of $\nu_{u}$ is proportional
to $e^{-n\,\norm x^{2}/2}\,h_{u}(x)$, where $h_{u}(x):=\int e^{-ut}\,\ind_{\bar{\K}}(x,t)\,\D t$
is logconcave. By the reduction above, $\{0\}\times[a,a+1]\subset\bar{\K}\subset\R^{n}\times[a-10,a+b]$.
Hence, for $u\in[0,n]$, 
\begin{align*}
h_{u}(0) & \ge\int^{a+1}_{a}e^{-ut}\,\D t\ge e^{-u(a+1)}\,,\\
h_{u}(x) & \le\int^{a+b}_{a-10}e^{-ut}\,\D t\le(b+10)\,e^{-u(a-10)}\,.
\end{align*}

Let $x_{u}$ be a mode of this marginal. Since $\nu^{X}_{u}(x_{u})/\nu^{X}_{u}(0)\geq1$,
\[
\frac{n}{2}\,\norm{x_{u}}^{2}\le\log\frac{h_{u}(x_{u})}{h_{u}(0)}\le\log(b+10)+11u\lesssim n\,.
\]
Also, since the $X$-marginal of $\nu_{u}$ is $n$-strongly logconcave,
the standard mode-moment estimate gives $\E_{\nu_{u}}[\norm{X-x_{u}}^{2}]\le1$
(see \cite[Lemma 4.0.1]{chewi25log}). Combining these bounds gives
$\E_{\nu_{u}}[\norm X^{2}]\lesssim1$. Now, by \eqref{eq:pi}, for
$Q(x)=n\,\norm x^{2}/2$, 
\[
\var_{\nu_{u}}Q\le\frac{1}{n}\,\E_{\nu_{u}}[\norm{\nabla Q}^{2}]=n\,\E_{\nu_{u}}[\norm X^{2}]\lesssim n\,.
\]
By the varentropy inequality~\cite{FMW16varentropy}, we have $\var_{\nu_{u}}(Q(X)+uT)\le n+1$.
Using the $L^{2}$-triangle inequality, 
\[
u\sqrt{\var_{\nu_{u}}T}\le\sqrt{\var_{\nu_{u}}\bpar{Q(X)+uT}}+\sqrt{\var_{\nu_{u}}Q}\lesssim\sqrt{n}\,.
\]
Hence, $\Phi''(u)=\var_{\nu_{u}}T\lesssim n/u^{2}\leq n/s^{2}_{-}$
for every $u\in I$, which gives $V_{I}\lesssim n/[s^{2}\,(1-(q-1)\,\alpha)^{2}]$.
Therefore,
\[
\msf R_{q}(\nu_{s}\mmid\nu_{(1+\alpha)\,s})\lesssim\frac{qn\alpha^{2}}{\bpar{1-(q-1)\,\alpha}^{2}}\,.
\]
This completes the proof.
\end{proof}

In the second phase, we keep $\rho=n$ fixed and increase $\sigma^{2}$.
While $\sigma^{2}<1$, we use the universal Gaussian-annealing update
$\sigma^{2}\gets\sigma^{2}(1+O(\frac{1}{(qn)^{1/2}}))$. Once $\sigma^{2}\ge1$,
we use the thin-shell update $\sigma^{2}\gets\sigma^{2}\,(1+O(\frac{\sigma^{2}}{(q\msf V)^{1/2}}))$
as in \S\ref{sec:uniform}. Let $\mu_{\sigma^{2}}:=\mu_{\sigma^{2},n}$,
and note that its $X$-marginal is $\bar{\pi}_{\sigma^{2}}$. The
conditional distribution $T|X=x$ is independent of $\sigma^{2}$,
so
\[
\msf R_{q}(\mu_{\sigma^{2}}\mmid\mu_{\widetilde{\sigma}^{2}})=\msf R_{q}(\bar{\pi}_{\sigma^{2}}\mmid\bar{\pi}_{\widetilde{\sigma}^{2}})\qquad\text{for all }q>1\ \text{and}\ \sigma^{2},\widetilde{\sigma}^{2}>0\,.
\]
Hence, the closeness bound in Corollary~\ref{cor:renyi-t} applies
in $\sigma^{2}\geq1$ and at the terminal point, but the regime of
$\sigma^{2}\leq1$ needs an additional argument.
\begin{lem}
\label{lem:sigma-closeness-phase1} Let $\mu_{s}:=\mu_{s,n}$ for
$s\in(0,1)$. Then, for $q>1$ and $\alpha>0$,
\[
\msf R_{q}(\mu_{s}\mmid\mu_{(1+\alpha)\,s})\lesssim qn\alpha^{2}\,.
\]
\end{lem}

\begin{proof}
Let $x_{\sigma}$ be a mode of $\bar{\pi}_{\sigma^{2}}$. The bounds
on $h_{n}$ in the preceding proof give $\norm{x_{\sigma}}^{2}/(2\sigma^{2})\leq\log(h_{n}(x_{\sigma})/h_{n}(0))\lesssim n$.
Since $\bar{\pi}_{\sigma^{2}}$ is $\sigma^{-2}$-strongly logconcave,
the same mode-moment estimate gives $\E_{\bar{\pi}_{\sigma^{2}}}[\norm{X-x_{\sigma}}^{2}]\leq n\sigma^{2}$.
Thus, $\E_{\bar{\pi}_{\sigma^{2}}}[\norm X^{2}]\lesssim n\sigma^{2}$,
so by \eqref{eq:pi}
\[
\var_{\bar{\pi}_{\sigma^{2}}}(\norm X^{2})\leq4\sigma^{2}\,\E_{\bar{\pi}_{\sigma^{2}}}[\norm X^{2}]\lesssim n\sigma^{4}\,.
\]
By Lemma~\ref{lem:renyi-exact} with $s=(2\sigma^{2})^{-1}$ and
$s'=s/(1+\alpha)$,
\[
\msf R_{q}(\mu_{\sigma^{2}}\mmid\mu_{(1+\alpha)\,\sigma^{2}})\lesssim qn\,\bpar{\frac{s-s'}{s'}}^{2}=qn\alpha^{2}\,,
\]
which justifies the update when $\sigma^{2}\leq1$.
\end{proof}

\subsection{Gaussian cooling in the lifted space}

Fix any target R\'enyi order $q\ge q_{0}$. We use $\psann$ to follow
the annealing scheme described above \cite{KV25faster,KV26zeroLC}.
Starting from $v_{0}\sim\mu_{0}$, with step size $h>0$ and threshold
$\tau$, write $v_{k}=(x_{k},t_{k})$, $w_{k+1}=(y_{k+1},s_{k+1})$,
$r:=\sigma^{2}/(h+\sigma^{2})$, and $h_{r}:=rh$. One iteration of
$\psann$ consists of the following two steps.
\begin{itemize}
\item {[}Forward{]} Draw $w_{k+1}\sim\msf N(v_{k},hI_{n+1})$.
\item {[}Backward{]} Repeatedly draw $v_{k+1}\sim\msf N(ry_{k+1},h_{r}I_{n})\otimes\msf N(s_{k+1}-\rho h,h)$
until $v_{k+1}\in\bar{\K}$. If the rejection loop exceeds $\tau$
trials, declare failure.
\end{itemize}
We use the following per-target guarantee from \cite[Proposition 5.2]{KV26zeroLC}.
For the cited guarantee, translate $t$ by $-(a+11)$. The shifted
body contains $B_{n+1}(0,1)$ and lies in $\R^{n}\times[-21,b-11]$,
while the normalized targets and both Gaussian steps transform by
the same translation. Thus, rejection counts and failure events are
unchanged, and the guarantee has no dependence on $a$.
\begin{prop}
[Lifted annealing sampler]\label{prop:lifted-relay} Let $\mu=\mu_{\sigma^{2},\rho}$
be as in~\eqref{eq:two-parameter-path} with $\sigma^{2}\geq1/n$
and $1\leq\rho\leq n$. Given access to an evaluation oracle for $V$,
let $\varepsilon>0$, $\eta\in(0,1/2)$, $q\ge q_{0}\ge2$, and let
$\mu_{0}$ be an initial distribution with $M_{q_{0}}:=\norm{\frac{\D\mu_{0}}{\D\mu}}_{L^{q_{0}}(\mu)}$.
Initialize $\psann$ from $\mu_{0}$ and iterate it $N=\Otilde(n^{2}(\sigma^{2}\vee1)\log\frac{q}{\varepsilon}\log\frac{1}{\eta})$
times with $h=(24^{2}n^{2}\log S)^{-1}$ and $\tau=2S^{2}\log^{2}S$,
where $S=16NM_{2}/\eta$. Assume $M_{q_{0}}\le10$ and $q_{0}\ge1+\log\tau$.
\begin{itemize}
\item The latter condition is ensured by $q_{0}\ge2\vee O(\log\frac{N}{\eta})$. 
\item With probability at least $1-\eta$, $\psann$ completes all $N$
iterations without failure. Conditioned on this event, its output
law $\nu$ satisfies $\msf R_{q}(\nu\mmid\mu)\le\varepsilon+2\log\frac{1}{1-\eta}$,
using $\Otilde(n^{2}\,(\sigma^{2}\vee1)\log\frac{q}{\varepsilon}\log^{2}\frac{1}{\eta})$
evaluation queries in expectation. The corresponding uncapped chain
(i.e., $\tau=\infty$) has output law $\bar{\nu}$ satisfying $\msf R_{q}(\bar{\nu}\mmid\mu)\le\varepsilon$. 
\end{itemize}
\end{prop}

Let $J$ be the total number of phases. Note that with $\varepsilon_{i}=1/20$
and $\eta_{i}=(100J)^{-1}$, each phase requires $\Otilde(n^{2}\,(\sigma^{2}\vee1))$
evaluation queries, and its desired output is $1/20$-close to its
target.
\begin{lyxalgorithm}
[Thin-shell Gaussian cooling]\label{alg:lifted-annealing} \textbf{\emph{Input:}}
evaluation oracle for $V$, a target order $q$, $\msf V\lesssim_{\log n}R^{2}L\wedge R^{3}\sqrt{\Lambda}$
, $\sigma^{2}_{\textup{last}}=\sqrt{q\msf V}$, and the base order
$q_{0}=2\vee\widetilde{\Theta}(\log(qn^{3}\cV))$, and a sufficiently
small constant $c>0$.

Every sampler call below uses the indicated input and output R\'enyi
orders, accuracy $\varepsilon_{i+1}=1/20$, failure probability $\eta_{i+1}=(100J)^{-1}$,
and the parameters of Proposition~\ref{prop:lifted-relay}, using
$S_{i+1}:=160N_{i+1}/\eta_{i+1}$. After each successful phase, denote
its output by $Z_{i+1}$ and replace $i$ by $i+1$. Restart the entire
construction if a cap is hit. The steps are as follows: 
\begin{enumerate}
\item Independently draw $G\sim\msf N(0,n^{-1}I_{n})$ conditioned on $\norm G\le D$
and $U\sim\Exp(1)$ conditioned on $0\le U\le b+10$, set $T:=a-10+U$,
and repeat until $V(G)\le nT$. Set $Z_{0}:=(G,T)\sim\mu_{0}:=\mu_{n^{-1},1}$,
$(\sigma^{2}_{\textup{start}},\rho_{0}):=(n^{-1},1)$, $\sigma^{2}_{0}:=\sigma^{2}_{\textup{start}}$,
and $i=0$.
\item \emph{{[}Phase I{]}} While $\rho_{i}<n$ and $\sigma^{2}_{i}=n^{-1}$,
set $\sigma^{2}_{i+1}\gets n^{-1}$ and 
\begin{equation}
\rho_{i+1}\gets\min\bbrace{n,\bpar{1+\frac{c}{\sqrt{q_{0}(q_{0}+n)}}}\,\rho_{i}}\,,\label{eq:vertical-schedule}
\end{equation}
and sample from  $\mu_{i+1}:=\mu_{n^{-1},\rho_{i+1}}$ using $\psann:\msf R_{q_{0}}\to\msf R_{2q_{0}}$.
\item \emph{{[}Phase II{]}} First, while $\sigma^{2}_{i}<1$ and $\rho_{i}=n$,
set $\rho_{i+1}\gets n$ and 
\[
\sigma^{2}_{i+1}\gets\bpar{1+\frac{c}{\sqrt{q_{0}n}}}\,\sigma^{2}_{i}\,,
\]
and sample from $\mu_{i+1}:=\mu_{\sigma^{2}_{i+1},n}$ using $\psann:\msf R_{q_{0}}\to\msf R_{2q_{0}}$.
After this loop, set $\sigma^{2}_{i+1}:=\sigma^{2}_{i}$, $\rho_{i+1}:=\rho_{i}$,
and $\mu_{i+1}:=\mu_{i}$, and make one call at this target with $\psann:\msf R_{q_{0}}\to\msf R_{2q}$.
Then, while $\sigma^{2}_{i}<\sigma^{2}_{\textup{last}}$, set $\rho_{i+1}\gets n$
and 
\[
\sigma^{2}_{i+1}\gets\bpar{1+\frac{\sigma^{2}_{i}}{\sqrt{q\msf V}}}\,\sigma^{2}_{i}\,,
\]
and sample from $\mu_{i+1}:=\mu_{\sigma^{2}_{i+1},n}$ using $\psann:\msf R_{q}\to\msf R_{2q}$.
After the second loop, return the $X$-coordinate of $Z_{i}$.
\end{enumerate}
\end{lyxalgorithm}

For Phase I updates, Lemma~\ref{lem:t-closeness} ensures $\msf R_{2q_{0}}(\mu_{n^{-1},\rho_{i}}\mmid\mu_{n^{-1},\rho_{i+1}})\le\frac{1}{18}$.
In the first Phase II loop, Lemma~\ref{lem:sigma-closeness-phase1}
ensures $1/4$-closeness in $\msf R_{2q_{0}}$-divergence. In the
second loop, as in \S\ref{sec:uniform}, Corollary~\ref{cor:renyi-t}
ensures $1/4$-closeness in $\msf R_{2q}$.
\begin{proof}
[Proof of Theorem~\ref{thm:warm-main}] We may assume $q\geq q_{0}$;
otherwise, we can proceed with $q=q_{0}$, and use monotonicity of
$\msf R_{q}$. Let $\mu_{0},\ldots,\mu_{J}$ denote the annealing
distributions in Algorithm~\ref{alg:lifted-annealing}, and set $\bar{\nu}_{0}:=\mu_{0}$.
For $i\ge1$, let $Z_{i}\sim\bar{\nu}_{i}$ be the output of the ideal
$\psann$ at phase $i$.

Let $p=q_{0}$ before the order-boost and $p=q$ after it. By the
weak triangle inequality, 
\[
\msf R_{p}(\bar{\nu}_{i-1}\mmid\mu_{i})\le\frac{p-1/2}{p-1}\,\msf R_{2p}(\bar{\nu}_{i-1}\mmid\mu_{i-1})+\msf R_{2p-1}(\mu_{i-1}\mmid\mu_{i})\le\frac{3}{40}+\frac{1}{4}=\frac{13}{40}\,.
\]
Hence, Proposition~\ref{prop:lifted-relay} with $(p,r)=(q_{0},2q_{0})$
before the boost and $(p,r)=(q,2q)$ afterward preserves $M_{p}$-warmness.

As in the proof of Theorem~\ref{thm:warm-uniform}, we couple the
capped and ideal paths until a cap is hit; the success event $S$
satisfies $\P(S)\ge99/100$. Write $\nu^{X,T}=\law(Z_{J}\mid S)$.
The terminal bound and the weak triangle inequality yield $\msf R_{q}(\bar{\nu}_{J}\mmid\bar{\pi}^{X,T})\le13/40$.
Since $\D\nu^{X,T}/\D\bar{\nu}_{J}\le1/\P(S)$,
\[
\msf R_{q}(\nu^{X,T}\mmid\bar{\pi}^{X,T})\le\frac{13}{40}+\frac{q}{q-1}\log\frac{100}{99}<1\,.
\]
Since $\bar{\pi}^{X,T}(\D x,\D t)=p^{-1}\,\ind_{\bar{\K}}(x,t)\,\pi^{X,T}(\D x,\D t)$
for $p=\pi^{X,T}(\bar{\K})\ge1/2$, data processing under $(x,t)\mapsto x$
gives 
\[
\msf R_{q}(\nu^{X}\mmid\pi)\underset{\textup{DPI}}{\le}\msf R_{q}(\nu^{X,T}\mmid\pi^{X,T})=\msf R_{q}(\nu^{X,T}\mmid\bar{\pi}^{X,T})+\log\frac{1}{p}\le1+\log2<2\,.
\]

It remains to bound the query complexity. Rejection sampling gives
an exact initial sample and uses $O(1)$ expected evaluation queries.
In Phase I, the update in~\eqref{eq:vertical-schedule} requires
$O(\sqrt{q_{0}\,(q_{0}+n)}\log n)$ phases. Each phase requires $\Otilde(n^{2})$
queries, so 
\[
\text{Phase I cost}\leq\Otilde(n^{2})\times O(\sqrt{q_{0}\,(q_{0}+n)}\log n)=\Otilde(n^{5/2})\,.
\]
For Phase II, the universal updates below $\sigma^{2}=1$ use $O(\sqrt{q_{0}n}\log n)$
phases, and each phase needs $\Otilde(n^{2})$ queries. After $\sigma^{2}\ge1$,
since the doubling of an initial $\sigma^{2}$ requires $O(\sqrt{q\msf V}/\sigma^{2})$
phases, and each phase needs $\Otilde(n^{2}\sigma^{2})$ queries,
the doubling requires $\Otilde(n^{2}\sqrt{q\msf V})$ queries. Summing
the query costs of the two loops,
\[
\text{Phase II cost}\le\Otilde(n^{2.5}+n^{2}\sqrt{q\msf V})=\Otilde\bpar{n^{5/2}+n^{2}\sqrt{q}\,(R\sqrt{L}\wedge R^{3/2}\Lambda^{1/4})}\,.
\]
Finally, since $\P(S)\ge99/100$, the restart wrapper multiplies the
expected cost by at most $100/99$, which completes the proof.
\end{proof}

\paragraph{Acknowledgments.}

This work was supported in part by NSF Awards CCF-2504995, CCF-2236669,
CCF-2504994, and a Simons Investigator award.

\bibliographystyle{alpha}
\bibliography{main}

\end{document}

%% file: math-macros.tex
\input{_macros.tex}

\global\long\def\on#1{\operatorname{#1}}%

\global\long\def\bw{\mathsf{Ball\ walk}}%
\global\long\def\sw{\mathsf{Speedy\ walk}}%
\global\long\def\gw{\mathsf{Gaussian\ walk}}%
\global\long\def\ps{\mathsf{Proximal\ sampler}}%
\global\long\def\dw{\mathsf{Dikin\ walk}}%

\global\long\def\chr{\mathsf{Coordinate\ Hit\text{-}and\text{-}Run}}%
\global\long\def\har{\mathsf{Hit\text{-}and\text{-}Run}}%
\global\long\def\gc{\mathsf{Gaussian\ cooling}}%
\global\long\def\ino{\mathsf{\mathsf{In\text{-}and\text{-}Out}}}%
\global\long\def\tgc{\mathsf{Tilted\ Gaussian\ cooling}}%
\global\long\def\PS{\mathsf{PS}}%
\global\long\def\psunif{\mathsf{PS}_{\textup{unif}}}%
\global\long\def\psexp{\mathsf{PS}_{\textup{exp}}}%
\global\long\def\psann{\mathsf{PS}_{\textup{ann}}}%
\global\long\def\psgauss{\mathsf{PS}_{\textup{Gauss}}}%
\global\long\def\eval{\mathsf{Eval}}%
\global\long\def\mem{\mathsf{Mem}}%

\global\long\def\O{O}%
\global\long\def\Otilde{\widetilde{O}}%
\global\long\def\Omtilde{\widetilde{\Omega}}%

\global\long\def\E{\mathbb{E}}%
\global\long\def\Z{\mathbb{Z}}%
\global\long\def\P{\mathbb{P}}%
\global\long\def\N{\mathbb{N}}%

\global\long\def\R{\mathbb{R}}%
\global\long\def\Rd{\mathbb{R}^{d}}%
\global\long\def\Rdd{\mathbb{R}^{d\times d}}%
\global\long\def\Rn{\mathbb{R}^{n}}%
\global\long\def\Rnn{\mathbb{R}^{n\times n}}%

\global\long\def\psd{\mathbb{S}^{d}_{+}}%
\global\long\def\pd{\mathbb{S}^{d}_{++}}%

\global\long\def\defeq{\stackrel{\mathrm{{\scriptscriptstyle def}}}{=}}%

\global\long\def\veps{\varepsilon}%
\global\long\def\lda{\lambda}%
\global\long\def\vphi{\varphi}%
\global\long\def\K{\mathcal{K}}%

\global\long\def\half{\frac{1}{2}}%
\global\long\def\nhalf{\nicefrac{1}{2}}%
\global\long\def\texthalf{{\textstyle \frac{1}{2}}}%
\global\long\def\ltwo{L^{2}}%

\global\long\def\ind{\mathds{1}}%
\global\long\def\op{\mathsf{op}}%
\global\long\def\ch{\mathsf{Ch}}%
\global\long\def\kls{\mathsf{KLS}}%
\global\long\def\ts{\mathsf{Ts}}%
\global\long\def\hs{\textup{HS}}%
\global\long\def\ls{\textup{LS}}%

\global\long\def\cpi{C_{\mathsf{PI}}}%
\global\long\def\cipi{C_{\mathsf{IPI}}}%
\global\long\def\clsi{C_{\mathsf{LSI}}}%
\global\long\def\cch{C_{\mathsf{Ch}}}%
\global\long\def\clch{C_{\mathsf{logCh}}}%
\global\long\def\cexp{C_{\mathsf{exp}}}%
\global\long\def\cgauss{C_{\mathsf{Gauss}}}%

\global\long\def\chooses#1#2{_{#1}C_{#2}}%

\global\long\def\frob{\on F}%

\global\long\def\vol{\on{vol}}%

\global\long\def\sym{\on{sym}}%

\global\long\def\law{\on{law}}%

\global\long\def\tr{\on{tr}}%

\global\long\def\diag{\on{diag}}%

\global\long\def\diam{\on{diam}}%

\global\long\def\poly{\on{poly}}%

\global\long\def\polylog{\on{polylog}}%

\global\long\def\Diag{\on{Diag}}%

\global\long\def\inter{\on{int}}%

\global\long\def\esssup{\on{ess\,sup}}%

\global\long\def\proj{\on{Proj}}%

\global\long\def\e{\mathrm{e}}%

\global\long\def\id{\mathrm{id}}%

\global\long\def\supp{\on{supp}}%

\global\long\def\spanning{\on{span}}%

\global\long\def\rows{\on{row}}%

\global\long\def\cols{\on{col}}%

\global\long\def\rank{\on{rank}}%

\global\long\def\T{\mathsf{T}}%

\global\long\def\bs#1{\boldsymbol{#1}}%

\global\long\def\eu#1{\EuScript{#1}}%

\global\long\def\mb#1{\mathbf{#1}}%

\global\long\def\mbb#1{\mathbb{#1}}%

\global\long\def\mc#1{\mathcal{#1}}%

\global\long\def\mf#1{\mathfrak{#1}}%

\global\long\def\ms#1{\mathscr{#1}}%

\global\long\def\mss#1{\mathsf{#1}}%

\global\long\def\msf#1{\mathsf{#1}}%

\global\long\def\cQ{\msf Q}%
\global\long\def\cV{\msf V}%
\global\long\def\cC{\msf C}%

\global\long\def\textint{{\textstyle \int}}%
\global\long\def\Dd{\mathrm{D}}%
\global\long\def\D{\mathrm{d}}%
\global\long\def\grad{\nabla}%
 
\global\long\def\hess{\nabla^{2}}%
 
\global\long\def\lapl{\triangle}%
 
\global\long\def\deriv#1#2{\frac{\D#1}{\D#2}}%
 
\global\long\def\pderiv#1#2{\frac{\partial#1}{\partial#2}}%
 
\global\long\def\de{\partial}%
\global\long\def\lagrange{\mathcal{L}}%
\global\long\def\Div{\on{div}}%

\global\long\def\Gsn{\mathcal{N}}%
 
\global\long\def\BeP{\textnormal{BeP}}%
 
\global\long\def\Ber{\textnormal{Ber}}%
 
\global\long\def\Bern{\textnormal{Bern}}%
 
\global\long\def\Bet{\textnormal{Beta}}%
 
\global\long\def\Beta{\textnormal{Beta}}%
 
\global\long\def\Bin{\textnormal{Bin}}%
 
\global\long\def\BP{\textnormal{BP}}%
 
\global\long\def\Dir{\textnormal{Dir}}%
 
\global\long\def\DP{\textnormal{DP}}%
 
\global\long\def\Exp{\textnormal{Exp}}%
 
\global\long\def\Gam{\textnormal{Gamma}}%
 
\global\long\def\GEM{\textnormal{GEM}}%
 
\global\long\def\HypGeo{\textnormal{HypGeo}}%
 
\global\long\def\Mult{\textnormal{Mult}}%
 
\global\long\def\NegMult{\textnormal{NegMult}}%
 
\global\long\def\Poi{\textnormal{Poi}}%
 
\global\long\def\Pois{\textnormal{Pois}}%
 
\global\long\def\Unif{\textnormal{Unif}}%

\global\long\def\bpar#1{\bigl(#1\bigr)}%
\global\long\def\Bpar#1{\Bigl(#1\Bigr)}%

\global\long\def\abs#1{|#1|}%
\global\long\def\babs#1{\bigl|#1\bigr|}%
\global\long\def\Babs#1{\Bigl|#1\Bigr|}%

\global\long\def\snorm#1{\|#1\|}%
\global\long\def\bnorm#1{\bigl\Vert#1\bigr\Vert}%
\global\long\def\Bnorm#1{\Bigl\Vert#1\Bigr\Vert}%

\global\long\def\sbrack#1{[#1]}%
\global\long\def\bbrack#1{\bigl[#1\bigr]}%
\global\long\def\Bbrack#1{\Bigl[#1\Bigr]}%

\global\long\def\sbrace#1{\{#1\}}%
\global\long\def\bbrace#1{\bigl\{#1\bigr\}}%
\global\long\def\Bbrace#1{\Bigl\{#1\Bigr\}}%

\global\long\def\Abs#1{\left\lvert #1\right\rvert }%
\global\long\def\Par#1{\left(#1\right)}%
\global\long\def\Brack#1{\left[#1\right]}%
\global\long\def\Brace#1{\left\{  #1\right\}  }%

\global\long\def\inner#1{\langle#1\rangle}%
 
\global\long\def\binner#1#2{\left\langle {#1},{#2}\right\rangle }%

\global\long\def\norm#1{\lVert#1\rVert}%
\global\long\def\onenorm#1{\norm{#1}_{1}}%
\global\long\def\twonorm#1{\norm{#1}_{2}}%
\global\long\def\infnorm#1{\norm{#1}_{\infty}}%
\global\long\def\fronorm#1{\norm{#1}_{\text{F}}}%
\global\long\def\nucnorm#1{\norm{#1}_{*}}%
\global\long\def\staticnorm#1{\|#1\|}%
\global\long\def\statictwonorm#1{\staticnorm{#1}_{2}}%

\global\long\def\mmid{\mathbin{\|}}%

\global\long\def\otilde#1{\widetilde{O}(#1)}%
\global\long\def\wtilde{\widetilde{W}}%
\global\long\def\wt#1{\widetilde{#1}}%

\global\long\def\KL{\msf{KL}}%
\global\long\def\dtv{d_{\textrm{\textup{TV}}}}%
\global\long\def\FI{\msf{FI}}%
\global\long\def\tv{\msf{TV}}%
\global\long\def\TV{\msf{TV}}%

\global\long\def\cov{\on{cov}}%
\global\long\def\var{\on{Var}}%
\global\long\def\ent{\on{Ent}}%

\global\long\def\cred#1{\textcolor{red}{#1}}%
\global\long\def\cblue#1{\textcolor{blue}{#1}}%
\global\long\def\cgreen#1{\textcolor{green}{#1}}%
\global\long\def\ccyan#1{\textcolor{cyan}{#1}}%
\global\long\def\yk#1{\textcolor{red}{\textsf{[YK: #1]}}}%
\global\long\def\yb#1{\textcolor{blue}{\textsf{[yb: #1]}}}%

\global\long\def\iff{\Leftrightarrow}%
 
\global\long\def\textfrac#1#2{{\textstyle \frac{#1}{#2}}}%

\global\long\def\Expo{\textnormal{Expo}}%
\global\long\def\Tr{\on{Tr}}%
\global\long\def\onu{\bar{\nu}}%
\global\long\def\intk{\inter\K}%
\global\long\def\ncal{\mathcal{N}}%
\global\long\def\svec{\operatorname{svec}}%
\global\long\def\tvec{\operatorname{vec}}%
\global\long\def\del{\partial}%
\global\long\def\ovec{\operatorname{ovec}}%

\global\long\def\Sph{\mathbb{S}}%
\global\long\def\gap{\operatorname{gap}}%
\global\long\def\Id{\mathrm{Id}}%
\global\long\def\HR{\mathrm{HR}}%
\global\long\def\CHAR{\mathrm{CHAR}}%
\global\long\def\logp{\log_{+}}%

%% file: _macros.tex
%--------------------------------------------------------------------------------------------------------------------------------
% Environment shortcuts
%--------------------------------------------------------------------------------------------------------------------------------
\def\balign#1\ealign{\begin{align}#1\end{align}}
\def\baligns#1\ealigns{\begin{align*}#1\end{align*}}
\def\balignat#1\ealign{\begin{alignat}#1\end{alignat}}
\def\balignats#1\ealigns{\begin{alignat*}#1\end{alignat*}}
\def\bitemize#1\eitemize{\begin{itemize}#1\end{itemize}}
\def\benumerate#1\eenumerate{\begin{enumerate}#1\end{enumerate}}

% Align environments that use textstyle instead of displaystyle
\newenvironment{talign*}
 {\let\displaystyle\textstyle\csname align*\endcsname}
 {\endalign}
\newenvironment{talign}
 {\let\displaystyle\textstyle\csname align\endcsname}
 {\endalign}

\def\balignst#1\ealignst{\begin{talign*}#1\end{talign*}}
\def\balignt#1\ealignt{\begin{talign}#1\end{talign}}
%---------------------------------------------------

%--------------------------------------------------------------------------------------------------------------------------------
% Redefine left and right to remove initial and trailing space
%--------------------------------------------------------------------------------------------------------------------------------
\let\originalleft\left
\let\originalright\right
\renewcommand{\left}{\mathopen{}\mathclose\bgroup\originalleft}
\renewcommand{\right}{\aftergroup\egroup\originalright}

%--------------------------------------------------------------------------------------------------------------------------------
% Words with special symbols
%--------------------------------------------------------------------------------------------------------------------------------
\def\Gronwall{Gr\"onwall\xspace}
\def\Holder{H\"older\xspace}
\def\Ito{It\^o\xspace}
\def\Nystrom{Nystr\"om\xspace}
\def\Schatten{Sch\"atten\xspace}
\def\Matern{Mat\'ern\xspace}

%--------------------------------------------------------------------------------------------------------------------------------
% Smaller citations
%--------------------------------------------------------------------------------------------------------------------------------
\def\tinycitep*#1{{\tiny\citep*{#1}}}
\def\tinycitealt*#1{{\tiny\citealt*{#1}}}
\def\tinycite*#1{{\tiny\cite*{#1}}}
\def\smallcitep*#1{{\scriptsize\citep*{#1}}}
\def\smallcitealt*#1{{\scriptsize\citealt*{#1}}}
\def\smallcite*#1{{\scriptsize\cite*{#1}}}

%--------------------------------------------------------------------------------------------------------------------------------
% Colors
%--------------------------------------------------------------------------------------------------------------------------------
\def\blue#1{\textcolor{blue}{{#1}}}
\def\green#1{\textcolor{green}{{#1}}}
\def\orange#1{\textcolor{orange}{{#1}}}
\def\purple#1{\textcolor{purple}{{#1}}}
\def\red#1{\textcolor{red}{{#1}}}
\def\teal#1{\textcolor{teal}{{#1}}}

%--------------------------------------------------------------------------------------------------------------------------------
% Font styles
%--------------------------------------------------------------------------------------------------------------------------------
\def\mbi#1{\boldsymbol{#1}} % Bold and italic (math bold italic)
\def\mbf#1{\mathbf{#1}}
\def\mrm#1{\mathrm{#1}}
\def\tbf#1{\textbf{#1}}
\def\tsc#1{\textsc{#1}}

%--------------------------------------------------------------------------------------------------------------------------------
% Bold and italic variables
%--------------------------------------------------------------------------------------------------------------------------------
\def\mbiA{\mbi{A}}
\def\mbiB{\mbi{B}}
\def\mbiC{\mbi{C}}
\def\mbiDelta{\mbi{\Delta}}
\def\mbif{\mbi{f}}
\def\mbiF{\mbi{F}}
\def\mbih{\mbi{g}}
\def\mbiG{\mbi{G}}
\def\mbih{\mbi{h}}
\def\mbiH{\mbi{H}}
\def\mbiI{\mbi{I}}
\def\mbim{\mbi{m}}
\def\mbiP{\mbi{P}}
\def\mbiQ{\mbi{Q}}
\def\mbiR{\mbi{R}}
\def\mbiv{\mbi{v}}
\def\mbiV{\mbi{V}}
\def\mbiW{\mbi{W}}
\def\mbiX{\mbi{X}}
\def\mbiY{\mbi{Y}}
\def\mbiZ{\mbi{Z}}

%--------------------------------------------------------------------------------------------------------------------------------
% Textstyle vs. displaystyle
%--------------------------------------------------------------------------------------------------------------------------------
\def\textsum{{\textstyle\sum}} % Sum in textstyle form
\def\textprod{{\textstyle\prod}} % Prod in textstyle form
\def\textbigcap{{\textstyle\bigcap}} % Bigcap in textstyle form
\def\textbigcup{{\textstyle\bigcup}} % Bigcup in textstyle form

%--------------------------------------------------------------------------------------------------------------------------------
% Mathematical sets
%--------------------------------------------------------------------------------------------------------------------------------
\def\reals{\mathbb{R}} % Real number symbol
\def\integers{\mathbb{Z}} % Integer symbol
\def\rationals{\mathbb{Q}} % Rational numbers
\def\naturals{\mathbb{N}} % Natural numbers
\def\complex{\mathbb{C}} % Complex numbers

\def\what#1{\widehat{#1}}

\def\twovec#1#2{\left[\begin{array}{c}{#1} \\ {#2}\end{array}\right]}
\def\threevec#1#2#3{\left[\begin{array}{c}{#1} \\ {#2} \\ {#3} \end{array}\right]}
\def\nvec#1#2#3{\left[\begin{array}{c}{#1} \\ {#2} \\ \vdots \\ {#3}\end{array}\right]} % An n-vector with three arguments

%--------------------------------------------------------------------------------------------------------------------------------
% Eigenvalues
%--------------------------------------------------------------------------------------------------------------------------------
\def\maxeig#1{\lambda_{\mathrm{max}}\left({#1}\right)}
\def\mineig#1{\lambda_{\mathrm{min}}\left({#1}\right)}

%--------------------------------------------------------------------------------------------------------------------------------
% Operators
%--------------------------------------------------------------------------------------------------------------------------------
\def\Re{\operatorname{Re}} % Real part
\def\indic#1{\mbb{I}\left[{#1}\right]} % Indicator function
\def\logarg#1{\log\left({#1}\right)} % log with argument
\def\polylog{\operatorname{polylog}}
\def\maxarg#1{\max\left({#1}\right)} % max with argument
\def\minarg#1{\min\left({#1}\right)} % min with argument
\def\Earg#1{\E\left[{#1}\right]}
\def\Esub#1{\E_{#1}}
\def\Esubarg#1#2{\E_{#1}\left[{#2}\right]}
\def\bigO#1{\mathcal{O}\left(#1\right)} % big-oh notation
\def\littleO#1{o(#1)} % big-oh notation
\def\P{\mbb{P}} % Probability symbol
\def\Parg#1{\P\left({#1}\right)}
\def\Psubarg#1#2{\P_{#1}\left[{#2}\right]}
\def\Trarg#1{\Tr\left[{#1}\right]} % Trace with argument
\def\trarg#1{\tr\left[{#1}\right]} % trace with argument
\def\Var{\mrm{Var}} % Variance symbol
\def\Vararg#1{\Var\left[{#1}\right]}
\def\Varsubarg#1#2{\Var_{#1}\left[{#2}\right]}
\def\Cov{\mrm{Cov}} % Covariance symbol
\def\Covarg#1{\Cov\left[{#1}\right]}
\def\Covsubarg#1#2{\Cov_{#1}\left[{#2}\right]}
\def\Corr{\mrm{Corr}} % Covariance symbol
\def\Corrarg#1{\Corr\left[{#1}\right]}
\def\Corrsubarg#1#2{\Corr_{#1}\left[{#2}\right]}
\newcommand{\info}[3][{}]{\mathbb{I}_{#1}\left({#2};{#3}\right)} % Information symbol
\newcommand{\staticexp}[1]{\operatorname{exp}(#1)} % An exponential with parens that do not resize with input
\newcommand{\loglihood}[0]{\mathcal{L}} % log likelihood

% Copied from mathrsfs.sty

%--------------------------------------------------------------------------------------------------------------------------------
% Optimization macros
%--------------------------------------------------------------------------------------------------------------------------------
%\providecommand{\argmax}{\mathop\mathrm{arg max}} % Defining math symbols
%\providecommand{\argmin}{\mathop\mathrm{arg min}}
\providecommand{\arccos}{\mathop\mathrm{arccos}}
\providecommand{\dom}{\mathop\mathrm{dom}}
\providecommand{\diag}{\mathop\mathrm{diag}}
\providecommand{\tr}{\mathop\mathrm{tr}}
\providecommand{\card}{\mathop\mathrm{card}}
\providecommand{\sign}{\mathop\mathrm{sign}}
\providecommand{\conv}{\mathop\mathrm{conv}} % Convex hull
\def\rank#1{\mathrm{rank}({#1})}
\def\supp#1{\mathrm{supp}({#1})}

\providecommand{\minimize}{\mathop\mathrm{minimize}}
\providecommand{\maximize}{\mathop\mathrm{maximize}}
\providecommand{\subjectto}{\mathop\mathrm{subject\;to}}

\def\openright#1#2{\left[{#1}, {#2}\right)}

%--------------------------------------------------------------------------------------------------------------------------------
% Proof environments
%--------------------------------------------------------------------------------------------------------------------------------
\ifdefined\nonewproofenvironments\else
% The Theorems are numbered consecutively
% Lemmas are numbered by section, and observations, claims, facts, and 
% assumptions take their numbering. Propositions and definitions have their
% own numbering by section.
\ifdefined\ispres\else
% These conflict with Beamer definitions in pres mode
% \newtheorem{theorem}{Theorem}
% \newtheorem{lemma}[theorem]{Lemma}
% \newtheorem{corollary}[theorem]{Corollary}
% \newtheorem{definition}[theorem]{Definition}
% \newtheorem{fact}[theorem]{Fact}
% \renewenvironment{proof}{\noindent\textbf{Proof.}\hspace*{.3em}}{\qed \vspace{.1in}}
% \newenvironment{proof-sketch}{\noindent\textbf{Proof Sketch}
%   \hspace*{1em}}{\qed\bigskip\\}
% \newenvironment{proof-idea}{\noindent\textbf{Proof Idea}
%   \hspace*{1em}}{\qed\bigskip\\}
% \newenvironment{proof-of-lemma}[1][{}]{\noindent\textbf{Proof of Lemma {#1}}
%   \hspace*{1em}}{\qed\\}
%   \newenvironment{proof-of-proposition}[1][{}]{\noindent\textbf{Proof of Proposition {#1}}
%   \hspace*{1em}}{\qed\\}
% \newenvironment{proof-of-theorem}[1][{}]{\noindent\textbf{Proof of Theorem {#1}}
%   \hspace*{1em}}{\qed\\}
% \newenvironment{proof-attempt}{\noindent\textbf{Proof Attempt}
%   \hspace*{1em}}{\qed\bigskip\\}
% \newenvironment{proofof}[1]{\noindent\textbf{Proof of {#1}}
%   \hspace*{1em}}{\qed\bigskip\\}
 
% \newtheorem*{remark*}{Remark}
% \newenvironment{remark}{\noindent\textbf{Remark.}
%   \hspace*{0em}}{\smallskip}%\bigskip}
% \newenvironment{remarks}{\noindent\textbf{Remarks}
%   \hspace*{1em}}{\smallskip}
% \fi
% \newtheorem{observation}[theorem]{Observation}
% \newtheorem{proposition}[theorem]{Proposition}
% \newtheorem{claim}[theorem]{Claim}
% \newtheorem{assumption}{Assumption}
% \theoremstyle{definition}
% \newtheorem{example}[theorem]{Example}
%\renewcommand{\theassumption}{\Alph{assumption}} % Set counter for assumptions
                                                 % to be alphabetical
\fi
\fi
% Makes equation numbers have (1.1) style
% \numberwithin{equation}{section}
% \numberwithin{equation}{subsection}
\makeatletter
\@addtoreset{equation}{section}
\makeatother
\def\theequation{\thesection.\arabic{equation}}

\newcommand{\cmark}{\ding{51}}

\newcommand{\xmark}{\ding{55}}

%--------------------------------------------------------------------------------------------------------------------------------
% Equation environments
%--------------------------------------------------------------------------------------------------------------------------------
\newcommand{\eq}[1]{\begin{align}#1\end{align}}
\newcommand{\eqn}[1]{\begin{align*}#1\end{align*}}
\renewcommand{\Pr}[1]{\mathbb{P}\left( #1 \right)}
\newcommand{\Ex}[1]{\mathbb{E}\left[#1\right]}
%\newcommand{\var}[1]{\text{Var}\left(#1\right)}
%\newcommand{\ind}[1]{{\mathbbm{1}}_{\{ #1 \}} }
%\newcommand{\abs}[1]{\left|#1\right|}

%--------------------------------------------------------------------------------------------------------------------------------
% Comment environments
%--------------------------------------------------------------------------------------------------------------------------------
\newcommand{\matt}[1]{{\textcolor{Maroon}{[Matt: #1]}}}
\newcommand{\kook}[1]{{\textcolor{blue}{[Kook: #1]}}}
\definecolor{OliveGreen}{rgb}{0,0.6,0}
\newcommand{\sv}[1]{{\textcolor{OliveGreen}{[Santosh: #1]}}}

%% file: main.bib
@inproceedings{LV18logSoblev,
	address = {New York, NY, USA},
	author = {Lee, Yin Tat and Vempala, Santosh S.},
	booktitle = {{S}ymposium on {T}heory of {C}omputing},
	pages = {1122--1129},
	publisher = {ACM},
	series = {STOC 2018},
	title = {Stochastic localization $+$ {S}tieltjes barrier $=$ tight bound for log-{S}obolev},
	year = {2018}}

@article{KV26HAR,
	author = {Kook, Yunbum and Vempala, Santosh S.},
	journal = {arXiv preprint arXiv:2608.16878},
	title = {Spectral gaps of {H}it-and-{R}un and {C}oordinate {H}it-and-{R}un},
	year = {2026}}

@article{CK26thinshell,
	author = {Chen, Yuansi and Klartag, Bo'az},
	journal = {arXiv preprint arXiv:2607.23307},
	title = {Digesting the proof of the sharp thin-shell inequality},
	year = {2026}}

@book{Kato95perturbation,
	address = {Berlin, Heidelberg},
	author = {Kato, Tosio},
	edition = {2},
	publisher = {Springer},
	series = {Classics in Mathematics},
	title = {Perturbation Theory for Linear Operators},
	year = {1995}}

@article{JLLV26reducing,
author = {Jia, He and Laddha, Aditi and Lee, Yin Tat and Vempala, Santosh},
title = {Reducing Isotropy and Volume to KLS: Faster Rounding and Volume Algorithms},
year = {2026},
issue_date = {April 2026},
publisher = {Association for Computing Machinery},
address = {New York, NY, USA},
volume = {73},
number = {2},
issn = {0004-5411},
url = {https://doi.org/10.1145/3795687},
doi = {10.1145/3795687},
journal = {J. ACM},
month = apr,
articleno = {8},
numpages = {21}
}

@article{KVZ26INO,
	author = {Kook, Yunbum and Vempala, Santosh S. and Zhang, Matthew S.},
	doi = {https://doi.org/10.1002/rsa.70061},
	eprint = {https://onlinelibrary.wiley.com/doi/pdf/10.1002/rsa.70061},
	journal = {Random Structures \& Algorithms},
	number = {3},
	pages = {e70061},
	title = {In-and-{O}ut: algorithmic diffusion for sampling convex bodies},
	url = {https://onlinelibrary.wiley.com/doi/abs/10.1002/rsa.70061},
	volume = {68},
	year = {2026}}

@inproceedings{KV25faster,
	author = {Kook, Yunbum and Vempala, Santosh S.},
	booktitle = {Symposium on Foundations of Computer Science},
	doi = {10.1109/FOCS63196.2025.00052},
	pages = {997-1006},
	publisher = {IEEE},
	title = {Faster logconcave sampling from a cold start in high dimension},
	year = {2025}}

@article{Bizeul26logsobolev,
	author = {Bizeul, Pierre},
	doi = {10.1016/j.jfa.2026.111368},
	issn = {0022-1236},
	journal = {Journal of Functional Analysis},
	month = may,
	number = {9},
	pages = {111368},
	publisher = {Elsevier BV},
	title = {On the log-{S}obolev constant of log-concave vectors},
	url = {http://dx.doi.org/10.1016/j.jfa.2026.111368},
	volume = {290},
	year = {2026}}

@article{KL25thin,
	author = {Klartag, Bo'az and Lehec, Joseph},
	journal = {arXiv preprint arXiv:2507.15495},
	title = {Thin-shell bounds via parallel coupling},
	year = {2025}}

@article{KV26zeroLC,
	archiveprefix = {arXiv},
	author = {Kook, Yunbum and Vempala, Santosh S.},
	doi = {10.48550/arXiv.2507.18021},
	eprint = {2507.18021},
	journal = {arXiv preprint arXiv:2507.18021},
	title = {Zeroth-order Logconcave Sampling},
	url = {https://arxiv.org/abs/2507.18021},
	year = {2025}}

@inproceedings{KV25sampling,
	author = {Kook, Yunbum and Vempala, Santosh S.},
	booktitle = {{S}ymposium on {T}heory of {C}omputing},
	doi = {10.1145/3717823.3718202},
	isbn = {979-8-4007-1510-5},
	mrclass = {68Q87},
	mrnumber = {4928485},
	pages = {924--932},
	publisher = {ACM},
	title = {Sampling and integration of logconcave functions by algorithmic diffusion},
	url = {https://doi.org/10.1145/3717823.3718202},
	year = {2025}}

@book{chewi25log,
	author = {Chewi, Sinho},
	publisher = {Book draft available at \url{https://chewisinho.github.io}},
	title = {Log-concave sampling},
	year = {2025}}

@inproceedings{Mironov17renyi,
	author = {Mironov, Ilya},
	booktitle = {Computer Security Foundations Symposium},
	doi = {10.1109/CSF.2017.11},
	issn = {2374-8303},
	pages = {263-275},
	publisher = {IEEE},
	title = {R\'{e}nyi differential privacy},
	url = {https://doi.ieeecomputersociety.org/10.1109/CSF.2017.11},
	year = {2017}}

@incollection{BK03central,
	author = {Bobkov, Sergey G. and Koldobsky, Alexander},
	booktitle = {Geometric Aspects of Functional Analysis},
	doi = {10.1007/978-3-540-36428-3\_5},
	mrclass = {60F05 (46B09 52A38)},
	mrnumber = {2083387},
	mrreviewer = {Sompong Dhompongsa},
	pages = {44--52},
	publisher = {Springer, Berlin},
	series = {Lecture Notes in Math.},
	title = {On the central limit property of convex bodies},
	url = {https://doi.org/10.1007/978-3-540-36428-3_5},
	volume = {1807},
	year = {2003}}

@article{ABP03central,
	author = {Anttila, Milla and Ball, Keith and Perissinaki, Irini},
	doi = {10.1090/S0002-9947-03-03085-X},
	fjournal = {Transactions of the American Mathematical Society},
	issn = {0002-9947},
	journal = {Transactions of the American Mathematical Society},
	mrclass = {52A22 (60F05)},
	mrnumber = {1997580},
	mrreviewer = {Erich Haeusler},
	number = {12},
	pages = {4723--4735},
	title = {The central limit problem for convex bodies},
	url = {https://doi.org/10.1090/S0002-9947-03-03085-X},
	volume = {355},
	year = {2003}}

@inproceedings{LST21structured,
	author = {Lee, Yin Tat and Shen, Ruoqi and Tian, Kevin},
	booktitle = {Conference on Learning Theory},
	pages = {2993--3050},
	publisher = {PMLR},
	title = {Structured logconcave sampling with a restricted {G}aussian oracle},
	url = {https://proceedings.mlr.press/v134/lee21a.html},
	volume = {134},
	year = {2021}}

@inproceedings{KZ25Renyi,
	author = {Kook, Yunbum and Zhang, Matthew S.},
	booktitle = {Symposium on Discrete Algorithms},
	doi = {10.1137/1.9781611978322.181},
	isbn = {9781611978322},
	pages = {5278--5306},
	publisher = {SIAM},
	title = {R\'{e}nyi-infinity constrained sampling with $d^3$ membership queries},
	url = {http://dx.doi.org/10.1137/1.9781611978322.181},
	year = {2025}}

@inproceedings{KVZ24INO,
	author = {Yunbum Kook and Santosh S. Vempala and Matthew S. Zhang},
	bibsource = {dblp computer science bibliography, https://dblp.org},
	booktitle = {Advances in Neural Information Processing Systems},
	pages = {108354--108388},
	publisher = {Curran Associates, Inc.},
	title = {In-and-{O}ut: algorithmic diffusion for sampling convex bodies},
	url = {http://papers.nips.cc/paper\_files/paper/2024/hash/c4006ff54a7bbda74c09bad6f7586f5b-Abstract-Conference.html},
	volume = {37},
	year = {2024}}

@inproceedings{LV06fast,
	author = {Lov\'{a}sz, L\'{a}szl\'{o} and Vempala, Santosh S.},
	booktitle = {Symposium on Foundations of Computer Science},
	doi = {10.1109/FOCS.2006.28},
	pages = {57-68},
	publisher = {IEEE},
	title = {Fast algorithms for logconcave functions: sampling, rounding, integration and optimization},
	year = {2006}}

@article{LV06hit,
	author = {Lov\'{a}sz, L\'{a}szl\'{o} and Vempala, Santosh S.},
	doi = {10.1137/s009753970544727x},
	issn = {1095-7111},
	journal = {SIAM Journal on Computing},
	number = {4},
	pages = {985--1005},
	publisher = {Society for Industrial and Applied Mathematics (SIAM)},
	title = {Hit-and-Run from a corner},
	url = {http://dx.doi.org/10.1137/S009753970544727X},
	volume = {35},
	year = {2006}}

@article{vH14renyi,
	author = {van Erven, Tim and Harremo{\"e}s, Peter},
	doi = {10.1109/TIT.2014.2320500},
	fjournal = {Institute of Electrical and Electronics Engineers. Transactions on Information Theory},
	issn = {0018-9448},
	journal = {IEEE Transactions on Information Theory},
	mrclass = {94A17},
	mrnumber = {3225930},
	mrreviewer = {S. M. Sunoj},
	number = {7},
	pages = {3797--3820},
	title = {R\'{e}nyi divergence and {K}ullback-{L}eibler divergence},
	url = {https://doi.org/10.1109/TIT.2014.2320500},
	volume = {60},
	year = {2014}}

@article{LV06simulated,
	author = {Lov\'{a}sz, L\'{a}szl\'{o} and Vempala, Santosh S.},
	doi = {10.1016/j.jcss.2005.08.004},
	fjournal = {Journal of Computer and System Sciences},
	issn = {0022-0000},
	journal = {Journal of Computer and System Sciences},
	mrclass = {68U05 (52A20 60C05 68W20)},
	mrnumber = {2205290},
	mrreviewer = {Carla Peri},
	number = {2},
	pages = {392--417},
	title = {Simulated annealing in convex bodies and an {$O^*(n^4)$} volume algorithm},
	url = {https://doi.org/10.1016/j.jcss.2005.08.004},
	volume = {72},
	year = {2006}}

@article{LS93random,
	author = {Lov\'{a}sz, L\'{a}szl\'{o} and Simonovits, Mikl\'{o}s},
	doi = {10.1002/rsa.3240040402},
	fjournal = {Random Structures \& Algorithms},
	issn = {1042-9832},
	journal = {Random Structures \& Algorithms},
	mrclass = {90C27 (52B55 90C15)},
	mrnumber = {1238906},
	mrreviewer = {Gerard Sierksma},
	number = {4},
	pages = {359--412},
	title = {Random walks in a convex body and an improved volume algorithm},
	url = {https://doi.org/10.1002/rsa.3240040402},
	volume = {4},
	year = {1993}}

@inproceedings{CV15bypass,
	author = {Cousins, Ben and Vempala, Santosh S.},
	booktitle = {{S}ymposium on {T}heory of {C}omputing},
	mrclass = {68W20 (52A38 60B10)},
	mrnumber = {3388233},
	pages = {539--548},
	publisher = {ACM},
	title = {Bypassing {KLS}: {G}aussian cooling and an {$O^*(n^3)$} volume algorithm},
	url = {https://mathscinet.ams.org/mathscinet-getitem?mr=3388233},
	year = {2015}}

@article{CV18Gaussian,
	author = {Cousins, Ben and Vempala, Santosh S.},
	doi = {10.1137/15M1054250},
	fjournal = {SIAM Journal on Computing},
	issn = {0097-5397},
	journal = {SIAM Journal on Computing},
	mrclass = {68W20 (52A38 60D05 65C40)},
	mrnumber = {3818340},
	number = {3},
	pages = {1237--1273},
	publisher = {Society for Industrial and Applied Mathematics (SIAM)},
	title = {Gaussian cooling and {$O^*(n^3)$} algorithms for volume and {G}aussian volume},
	url = {https://doi.org/10.1137/15M1054250},
	volume = {47},
	year = {2018}}

@article{LV24eldan,
	author = {Lee, Yin Tat and Vempala, Santosh S.},
	doi = {10.4007/annals.2024.199.3.2},
	fjournal = {Annals of Mathematics. Second Series},
	issn = {0003-486X},
	journal = {Annals of Mathematics},
	mrclass = {52A23 (60E15)},
	mrnumber = {4740210},
	number = {3},
	pages = {1043--1092},
	title = {Eldan's stochastic localization and the {KLS} conjecture: isoperimetry, concentration and mixing},
	url = {https://doi.org/10.4007/annals.2024.199.3.2},
	volume = {199},
	year = {2024}}

@article{KLS97random,
	author = {Kannan, Ravi and Lov\'{a}sz, L\'{a}szl\'{o} and Simonovits, Mikl\'{o}s},
	doi = {10.1002/(SICI)1098-2418(199708)11:1<1::AID-RSA1>3.0.CO;2-X},
	fjournal = {Random Structures \& Algorithms},
	issn = {1042-9832},
	journal = {Random Structures \& Algorithms},
	mrclass = {68Q25 (52A20 52A38 60J15 60J20)},
	mrnumber = {1608200},
	mrreviewer = {Mark R. Jerrum},
	number = {1},
	pages = {1--50},
	title = {Random walks and an {$O^*(n^5)$} volume algorithm for convex bodies},
	url = {https://doi.org/10.1002/(SICI)1098-2418(199708)11:1<1::AID-RSA1>3.0.CO;2-X},
	volume = {11},
	year = {1997}}

@article{DFK91random,
	author = {Dyer, Martin and Frieze, Alan and Kannan, Ravi},
	doi = {10.1145/102782.102783},
	fjournal = {Journal of the ACM},
	issn = {0004-5411},
	journal = {Journal of the ACM},
	mrclass = {68U05 (52B55 68Q25)},
	mrnumber = {1095916},
	number = {1},
	pages = {1--17},
	title = {A random polynomial-time algorithm for approximating the volume of convex bodies},
	url = {https://doi.org/10.1145/102782.102783},
	volume = {38},
	year = {1991}}

@article{Klartag23log,
	author = {Klartag, Bo'az},
	doi = {10.15781/jsjy-0b06},
	eprint = {2303.14938},
	fjournal = {Ars Inveniendi Analytica},
	journal = {Ars Inveniendi Analytica},
	mrclass = {52A40 (58J65)},
	mrnumber = {4603941},
	mrreviewer = {Ge Xiong},
	note = {Paper No. 4, 17 pp.},
	title = {Logarithmic bounds for isoperimetry and slices of convex sets},
	url = {https://doi.org/10.15781/jsjy-0b06},
	year = {2023}}

@article{KLS95isop,
	author = {Kannan, Ravi and Lov\'{a}sz, L\'{a}szl\'{o} and Simonovits, Mikl\'{o}s},
	doi = {10.1007/BF02574061},
	fjournal = {Discrete \& Computational Geometry. An International Journal of Mathematics and Computer Science},
	issn = {0179-5376},
	journal = {Discrete \& Computational Geometry},
	mrclass = {52A40 (52A38 68Q20)},
	mrnumber = {1318794},
	mrreviewer = {Carla Peri},
	number = {3-4},
	pages = {541--559},
	title = {Isoperimetric problems for convex bodies and a localization lemma},
	url = {https://doi.org/10.1007/BF02574061},
	volume = {13},
	year = {1995}}

@book{BGL14analysis,
	author = {Bakry, Dominique and Gentil, Ivan and Ledoux, Michel},
	doi = {10.1007/978-3-319-00227-9},
	isbn = {978-3-319-00226-2; 978-3-319-00227-9},
	mrclass = {60J25 (58J65 60J35 60J60)},
	mrnumber = {3155209},
	mrreviewer = {Ming Liao},
	pages = {xx+552},
	publisher = {Springer, Cham},
	title = {Analysis and geometry of {M}arkov diffusion operators},
	url = {https://doi.org/10.1007/978-3-319-00227-9},
	volume = {348},
	year = {2014}}

@incollection{FMW16varentropy,
	author = {Fradelizi, Matthieu and Madiman, Mokshay and Wang, Liyao},
	booktitle = {High Dimensional Probability VII},
	doi = {10.1007/978-3-319-40519-3_3},
	eprint = {1508.04093},
	pages = {45--60},
	publisher = {Birkh{\"a}user, Cham},
	series = {Progress in Probability},
	title = {Optimal Concentration of Information Content for Log-Concave Densities},
	url = {https://arxiv.org/abs/1508.04093},
	volume = {71},
	year = {2016}}

@article{CW01distributional,
	author = {Carbery, Anthony and Wright, James},
	doi = {10.4310/MRL.2001.v8.n3.a1},
	journal = {Mathematical Research Letters},
	number = {3},
	pages = {233--248},
	title = {Distributional and {$L^q$} Norm Inequalities for Polynomials over Convex Bodies in {$\mathbb{R}^n$}},
	volume = {8},
	year = {2001}}

@article{KV26unified,
	archiveprefix = {arXiv},
	author = {Kook, Yunbum and Vempala, Santosh S.},
	doi = {10.48550/arXiv.2606.12694},
	eprint = {2606.12694},
	journal = {arXiv preprint arXiv:2606.12694},
	title = {A Unified Complexity Bound for Logconcave Sampling},
	url = {https://arxiv.org/abs/2606.12694},
	year = {2026}}
